\documentclass[11pt]{amsart}
\usepackage{amsthm}

\usepackage[all]{xy}
\usepackage{amssymb}
\usepackage{enumerate}
\usepackage{mathrsfs}
\usepackage{hyperref}
\usepackage{epsfig}
\usepackage{graphicx}
\usepackage{subfig}
\usepackage{float}
\usepackage{epigraph}

\numberwithin{equation}{section}

\newtheorem{thm}{Theorem}[section]

\newtheorem{lem}[thm]{Lemma}

\newtheorem{rem}{Remark}[section]
\newtheorem{example}[thm]{Example}
\newtheorem{defin}[thm]{Definition}

\renewcommand{\Re}{\operatorname{\rm Re}}
\renewcommand{\Im}{\operatorname{\rm Im}}

\newcommand{\beqast}{\begin{eqnarray*}}
\newcommand{\eqast}{\end{eqnarray*}}
\newcommand{\beqa}{\begin{eqnarray}}
\newcommand{\eqa}{\end{eqnarray}}

\newcommand{\bbe}{\begin{equation}}
\newcommand{\ee}{\end{equation}}

\newcommand{\sbr}{\smallbreak}

\newcommand{\bfo}{{\bf 1}}

\newcommand{\bE}{{\mathbb E}}
\newcommand{\bP}{{\mathbb P}}
\newcommand{\bQ}{{\mathbb Q}}

\newcommand{\bR}{{\mathbb R}}
\newcommand{\bC}{{\mathbb C}}
\newcommand{\bZ}{{\mathbb Z}}

\newcommand{\cD}{{\mathcal D}}

\newcommand{\cL}{{\mathcal L}}

\newcommand{\cR}{{\mathcal R}}

\newcommand{\cC}{{\mathcal C}}

\newcommand{\eps}{\epsilon}
\newcommand{\de}{\delta}
\newcommand{\al}{\alpha}
\newcommand{\be}{\beta}

  \newcommand{\sg}{\sigma}

\newcommand{\De}{\Delta}
\newcommand{\la}{\lambda}
\newcommand{\lp}{\lambda_+}
\newcommand{\lm}{\lambda_-}
\newcommand{\La}{\Lambda}
\newcommand{\mum}{\mu_-}
\newcommand{\mup}{\mu_+}

\newcommand{\ka}{\kappa}

\newcommand{\om}{\omega}

\newcommand{\ze}{\zeta}

\newcommand{\ga}{\gamma}
\newcommand{\gap}{\gamma^+}
\newcommand{\gam}{\gamma^-}
\newcommand{\Ga}{\Gamma}

\begin{document}

\title
[SINH-acceleration]
{SINH-acceleration: efficient evaluation of probability distributions,
option pricing,  and Monte-Carlo simulations}

\author[
Svetlana Boyarchenko and
Sergei Levendorski\u{i}]
{
Svetlana Boyarchenko and
Sergei Levendorski\u{i}}

\thanks{
\emph{S.B.:} Department of Economics, The
University of Texas at Austin, 2225 Speedway Stop C3100, Austin,
TX 78712--0301, {\tt sboyarch@eco.utexas.edu} \\
\emph{S.L.:}
Calico Science Consulting. Austin, TX.
 Email address: {\tt
levendorskii@gmail.com}}

\begin{abstract}
Characteristic functions of several popular classes of distributions and processes
admit analytic continuation
into unions of strips and open coni around $\bR\subset \bC$.
The Fourier transform techniques reduces calculation of probability distributions and option prices to evaluation
of integrals whose integrands are analytic in domains enjoying these properties.
In the paper, we suggest to use
changes of variables of the form $\xi=\sqrt{-1}\om_1+b\sinh (\sqrt{-1}\om+y)$ and the simplified
trapezoid rule to evaluate the integrals accurately and fast. We formulate the general scheme, and
 apply the scheme for calculation
probability distributions and pricing European options in L\'evy models, the Heston model,
the CIR model, and a L\'evy model with the CIR-subordinator. We outline applications to fast and accurate calibration procedures
and Monte Carlo simulations in L\'evy models, regime switching L\'evy models that can account for stochastic drift, volatility and skewness,
and the Heston model. For calculation of quantiles in the tails using
the Newton or bisection method,  it suffices to precalculate several hundred of values of 
the characteristic exponent
at points of an appropriate grid ({\em conformal principal components}) and use these values in formulas for cpdf and pdf.
\end{abstract}
\maketitle


\section{Introduction}\label{intro}
In the paper, we formulate general conditions on
integrals arising in the Laplace and Fourier inversion,
the Wiener-Hopf factorization, calculation of probability distributions,  and pricing options and other derivative securities,
which make it possible to evaluate these integrals very accurately and fast. 
In applications to finance, these conditions are conditions on the characteristic functions of contingent claims. In the case of
European options in one-factor L\'evy models and some  popular affine models, the characteristic functions are functions defined on wide  
regions in the complex plane $\bC$; in the case of basket options, barrier and lookback options, the characteristic functions are defined on wide subsets of
$\bC^n$, where $n\ge 2$. An evaluation of the Wiener-Hopf factors, which appear in pricing formulas for barrier and lookback options, 
also involves functions on subsets of $\bC^n$, where $n\ge 2$. In the present paper, we consider several situations, where
one-dimensional integrals and functions on subsets of $\bC$ appear, and leave applications of the same techniques to multidimensional cases
to the future.

We start with the explanation of the main idea
of the suggested methodology in
 the case of one-dimensional integrals of the form
\begin{equation}\label{fxi}
I=\int_{\Im\xi=\om_0}e^{-ix\xi}g(\xi)d\xi,
\end{equation}
where $i=\sqrt{-1}$, $x\in \bR$, the line of integration $\{\Im\xi=\om_0\}$ is in the domain of analyticity of $g$, and $g(\xi)$
decays sufficiently fast as $\xi\to\infty$ remaining in the strip sandwiched between the lines $\bR$ and $\{\Im\xi=\om_0\}$. In probability, the simplest
integrals of this kind appear when  the characteristic function $g(\xi)=\bE[e^{i\xi Y}]$
of a random variable $Y$ is  well-defined not only on $\bR$ but on the line
$\{\Im\xi=\om_0\}$ as well. Then the RHS of \eqref{fxi}
 is the probability distribution function (pdf) $p_Y(x)$ of $Y$ evaluated at $x$ (and multiplied by $2\pi$).

\sbr
The first step of our methodology is the following change  the variable  in \eqref{fxi}:
\begin{equation}\label{sinhbasic}
\xi=\chi_{\om_1, \om; b}(y)=i\om_1+b\sinh (i\om+y),
\end{equation}
where $\om_0,\om_1\in \bR, \om\in [-\pi/2,\pi/2]$ and $b>0$ are related as follows: $\om_0=\om_1+b\sin (\om)$. We call
 the change of variables \eqref{sinhbasic}  the {\em sinh-acceleration} (in the case of multiple integrals, we make an appropriate
sinh-accelerations w.r.t. to  each argument). In the $y$- coordinate, we integrate over the real line.
The change of variables can be justified if the integrand $f(y)=e^{-ix \chi_{\om_1,\om; b}(y)}g(\chi_{\om_1,\om; b}(y))\chi_{\om_1,\om; b}'(y)$ admits analytic continuation
to a sufficiently wide strip $S_{(d_-,d_+)}=\{y\in \bC\ |\ \Im y\in (d_-,d_+)\}$ and decays sufficiently fast as $y\to\infty$ remaining
in the strip. In more detail, the Cauchy integral theorem allows us to deform the line of integration
$\{\Im\xi=\om_0\}$ into the contour $\cL_{\om_1,\om; b}:=\chi_{\om_1, \om; b}(\bR)$. In the integral over
$\cL_{\om_1, \om; b}$, we make the change of variables \eqref{sinhbasic}.
Thus, the sinh-acceleration is possible iff
the integrand  admits analytic continuation to the union of a strip $S$ around the line of integration
$\{\Im\xi=\om_0\}$ and an appropriate conus, and vanishes as $\xi\to\infty$ remaining in the conus. In some
important cases, analytic continuation to a wider region in an appropriate Riemann surface is possible, and then
the speed of the method improves.

We failed to invent a short name for a class of functions enjoying these properties: whereas the name
``functions analytic in a strip (and decaying at infinity)" is not exceedingly clumsy, the name ``functions analytic in a union of a strip and conus"
does seem clumsy. We suggest the name {\em sinh-regular functions}. We use the same adjective {\em sinh-regular} for distributions and processes
that lead to integrals of sinh-regular functions. For the same process, in different problems,
the sinh-acceleration with different sets of parameters $\om_1,\om, b$
needs to be used, therefore, we will formulate general conditions on the characteristic function of the process in terms of the strip
(in multi-factor models, tube domain) and conus
of analyticity, and list several wide classes of sinh-regular processes and distributions.

\sbr
The second step of our methodology is quite standard: the discretization of the integral using the infinite trapezoid rule.
If the integrand is analytic in a strip $S_{(\om_0-d,\om_0+d)}$ around the line of integration and decays sufficiently fast
as $\xi\to\infty$ remaining in the strip, the discretization error of the simplified trapezoid rule
decays exponentially as a function of the reciprocal to the mesh size $\ze$.
 Hence,
 a small error tolerance can be satisfied quite easily. Next, the infinite sum must be truncated; the resulting formula
 is called the simplified trapezoid rule:
 \[
 I=\ze \sum_{|j|\le N} f(j\ze).
 \]
As it is common in the literature,
 one can apply the simplifying trapezoid rule to the initial integral. However, in many cases of interest,
$g(\xi)$ decays slowly as $\xi\to\infty$ remaining in the strip, hence,
the truncation error decays slowly as the number of terms of the simplified trapezoid rule increases. For instance,
even a moderately accurate evaluation of the probability
distribution of a L\'evy process may require dozens of million of terms and more. The sinh-acceleration
exponentially increases the rate of decay of the integrand, and the number $N$ of terms sufficient to satisfy a given error significantly decreases.
In many cases, $N<10$ suffice to satisfy the error tolerance $\eps=10^{-7}$; typically, less than 50 terms suffice,
and in essentially all cases of interest, $N$ of the order of 100-150
suffices to satisfy the error tolerance $10^{-12}$.
\sbr

 A similar trick with the fractional-parabolic changes of the variables of the form
\begin{equation}\label{fractparabasic}
\xi=\chi^\pm_{\om;\sg; \al}(\eta)= i\om\pm i\sg(1\mp i\eta)^\al,
\end{equation}
where $\om\in\bR, \sg>0, \al>1$, was systematically used in a series of papers
\cite{iFT,paraHeston,paraLaplace,paired,MarcoDiscBarr,one-sidedCDS,pitfalls,BarrStIR,UltraFast,HestonCalibMarcoMe,HestonCalibMarcoMeRisk}.
In the working paper \cite{Sinh}, it was suggested to use the sinh-acceleration $\eta=\sinh(ay)/a$ with  integration over the real line after the fractional-parabolic change of variables
has been made. In the present paper, we use the sinh-acceleration only, in the more general form
\eqref{sinhbasic}. Depending on the sign of $\om$, the new contour of integration $\cL_{\om_1, \om, b}$ is deformed either upward
or downward, and the deformed contours enjoy properties similar to the properties of the contours in
\cite{iFT,paraHeston,paraLaplace,paired,MarcoDiscBarr,one-sidedCDS,pitfalls,BarrStIR,UltraFast,HestonCalibMarcoMe,HestonCalibMarcoMeRisk,Sinh}.
The number of terms of the simplified trapezoid rule
is approximately the same as in \cite{Sinh} (and smaller than in
\cite{iFT,paraHeston,paraLaplace,paired,MarcoDiscBarr,one-sidedCDS,pitfalls,BarrStIR,UltraFast,HestonCalibMarcoMe,HestonCalibMarcoMeRisk}, where thousands of terms
are needed in some cases)
but the number of elementary operations needed to calculate individual terms decreases. The general scheme is simpler than the one in \cite{Sinh}.
\sbr
The rest of the paper is organized as follows. In Section \ref{sRLPE}, we explain the sinh-acceleration techniques
in applications to evaluation of probability distribution functions of wide classes of L\'evy processes and infinitely
divisible distributions, which we call sinh-regular processes and distributions. The class contains almost
all popular classes of L\'evy processes used in finance, conditional probability distributions in the Heston model, more general stochastic volatility models,
affine and quadratic interest rate models, models with Wishart dynamics, Barndorf-Nielsen and Shephard model, 3/2 model,\ldots. As a basic numerical example, we
consider the probability distributions in
the NTS model \cite{B-N-L}.
In Section \ref{simpleII}, \ref{CIRbond} and \ref{subord}, we consider pricing European options in sinh-regular L\'evy models and the Heston model,
the CIR model, and the subordinated NTS model, the subordinator being the aggregated square root process.
Calculation of quantiles and applications to the Monte-Carlo simulations in L\'evy models, regime-switching L\'evy models and the Heston model
are outlined in Section \ref{s:MC}. For calculation of quantiles in  wide regions in the tails using
the Newton or bisection method,  it suffices to precalculate several hundred of values of functions that appear in the characteristic exponents
of the cpdf and pdf ({\em conformal principal components}), 
at points of an appropriate grid,  and use these values to evaluate  cpdf and pdf. 
Section \ref{concl} summarizes the results of the paper and outlines natural extensions. 
Tables are in the appendix.

\section{SINH-regular L\'evy processes and infinitely divisible distributions}\label{sRLPE}
\subsection{Definition}
In \cite{NG-MBS}, two almost equivalent definitions of a wide family of {\em Regular L\'evy Processes of exponential type} (RLPE)
are given: one in terms of the L\'evy density (exponential decay at infinity), the other one in terms of
the characteristic exponent (analytic in a strip around the real axis for processes
on $\bR$, and, for processes on $\bR^n$, in a tube domain $\bR^n+iU$, where $U\subset \bR^n$ is an open set containing 0). The class
of RLPEs contains all
classes of processes (model classes) popular in quantitative finance. The class of tempered stable L\'evy process as defined in \cite{Rosinski07}
is a subclass of RPLEs. In \cite{iFT,paraLaplace,paired}, it was noticed
that the characteristic exponents of processes
of the model classes admit analytic continuation to much wider regions of the complex plane and appropriate Riemann surfaces,
and enjoy several properties useful for the construction of  new efficient methods for pricing
contingent claims.
In this paper, we relax the general conditions of the definition of {\em strongly regular L\'evy processes of exponential type} (sRLPE)
introduced in \cite{paraLaplace}. Additional conditions formulated in \cite{paraLaplace,paired} were needed for the construction
 of more efficient methods when the Wiener-Hopf factorization and the fractional-parabolic change of variables were used.
 When the sinh-acceleration is used instead, the advantage of these additional conditions is marginal.

 Let $\lm<0<\lp$ and $-\pi/2\le \gam<\gap\le \pi/2$ and  either $\gam\le 0<\gap$ or
 $\gam<0\le \gap$. We define  the conus
 $\cC_{\gam,\gap}=\{e^{i\varphi}\rho\ |\ \rho\ge 0, \varphi\in (\gam,\gap)\cup (\pi-\gap,\pi-\gam)\}$, 
 and the set
 \bbe\label{defU}
 U(\lm,\lp; \gam, \gap)=i(\lm,\lp)+\cC_{\gam,\gap}:=\{ia+z\ |\ \lm<a<\lp, z\in \cC_{\gam,\gap}\}.
 \ee
As in \cite{paraLaplace}, we represent the characteristic exponent in the form
\begin{equation}\label{reprpsi}
\psi(\xi)=-i\mu\xi+\psi^0(\xi),
\end{equation}
and impose  conditions on $\psi^0$.
We need coni $\cC=\cC_{\gam,\gap}$ of several kind:
 \begin{enumerate}[(1)]
 \item
 to describe a domain $U=i(\lm,\lp)+\cC_{\gam,\gap}$ of analyticity of the characteristic exponent:
 \item
 to introduce a subset $U^u=i(\lm,\lp)+\cC^u$ where $|\psi^0(\xi)|$ admits a useful  upper bound;
 \item
 to introduce a subset $U^l=i(\lm,\lp)+\cC^l$ where $\Re\psi^0(\xi)$ admits a useful lower bound.
 \end{enumerate}
 In many cases, the coni are around the real axis. However, if $\psi^0(\xi)=o(|\xi|)$ as $\xi\to \infty$ in the domain of analyticity,
 and $\mu\neq 0$,
 then,   to calculate the pdf not at the peak and price options that are not at the money,
  we will have to choose a domain $U'=i(\lm,\lp)+\cC'$, where the conus $\cC'$ is either in the upper half-plane
 or low half-plane.\footnote{ Additional conditions on sets of the form \eqref{defU} are needed when the Wiener-Hopf factors are calculated;
 these conditions depend on the spectral parameter.}

\begin{defin}\label{defSINHLevy}
We say that $X$ is a SINH-regular L\'evy process  (on $\bR$) of type $((\mum,\mup);\cC; \cC_+)$ and order $\nu\in (0,2]$ iff
the following conditions are satisfied:
\begin{enumerate}[(i)]
\item
$\mum<0<\mup$;
\item
$\cC_+\subset \cC\subset\bC$ are open coni adjacent to or containing the real axis;
\item
$\psi$, the characteristic exponent of $X$, admits analytic continuation to
$i(\mum,\mup)+ \cC$;
\item
for any 
$\mum'\in (\mum,0)$, $\mup'\in (0, \mup)$ and an open sub-cone $\cC^u\subset \cC$ adjacent to or containing the real axis,
there exist $C>0$ such that 
\begin{equation}\label{boundpsisRLPE}
|\psi^0(\xi)|\le C(1+|\xi|)^\nu,\quad \forall\ \xi\in i[\mum',\mup']+ \cC^u;
\end{equation}
\item
for any closed  sub-cone $\cC^l_+\subset \cC_+$ and any $[\mum',\mup']\subset (\mum,\mup)$, there exist
 $c, C>0$ such that 
\begin{equation}\label{boundpsisRLPEp}
\Re\psi^0(\xi)\ge c|\xi|^\nu-C, \quad \forall\ \xi\in i[\mum',\mup']+\cC^l_+.
\end{equation}
\end{enumerate}
We say that a distribution is a SINH-regular infinitely divisible distribution of type $((\mum,\mup);\cC; \cC_+)$ and order $\nu$ iff it is the distribution
of $X_1$, where $X$ is a SINH-regular L\'evy process of type $((\mum,\mup);\cC; \cC_+)$ and order $\nu$.

\end{defin}
\begin{rem}{\rm We can generalize this definition allowing either $\mum=0$ or $\mup=0$ (but not both). However,
an appropriate Esscher transform can be used to reduce each of these two cases to the case $\mum<0<\mup$. The prominent example for
the case
$\mum=0=\mup$ are the stable L\'evy distributions and processes; we consider this case in a separate publication.

}
\end{rem}
\begin{defin}\label{defSINHLevyell}
We say that $X$ is an elliptic SINH-regular L\'evy process  of type $((\mum,\mup);\cC; \cC_+)$ and order $\nu$ if conditions (i)-(iii)
and the following two conditions hold
\vskip0.1cm\noindent
(iv)' for any $\xi_0\in \cC\cap \{\xi\ |\ |\xi|=1\}$, as $\rho\to+\infty$,
\begin{equation}\label{asympsisRLPE}
\psi^0(\rho \xi_0)\sim  c_\infty(\mathrm{arg}\,\xi_0)\rho^\nu,
\end{equation}
where $c_\infty(\mathrm{arg}\,\xi_0):=c_\infty(\psi^0, \mathrm{arg}\,\xi_0)$ is continuous;
\vskip0.1cm\noindent
(v)' for any $\xi_0\in \overline{\cC_+^l}\cap \{\xi\ |\ |\xi|=1\}$, $\Re c_\infty(\mathrm{arg}\,\xi_0)>0$.

We say that a distribution is an elliptic SINH-regular infinitely divisible distribution iff it is the distribution
of $X_1$, where $X$ is an elliptic SINH-regular L\'evy process.

\end{defin}

\begin{rem}{\rm
a) The properties needed for efficient calculations are formulated in the language of Complex Analysis,
and cannot be naturally formulated in the probabilistic language. Indeed, such simple processes as the Brownian motion with an embedded 
compound Poisson process with the L\'evy density $\bfo_{[a,b]}$, where $a<0<b$, are RLPEs and their characteristic exponents are analytic in
the complex plane but, on any conus, $\Re\psi^0(\xi)$ is not semi-bounded, and the crucial property \eqref{boundpsisRLPEp} fails.
However, if either $a=0, b>0$ or $a<0, b=0$, then such a process is an elliptic SINH-regular process of type $((\mum,\mup);\cC; \cC_+)$,
where $\mum<0<\mup$ are arbitrary, and $\cC, \cC_+$ are coni in the lower half plane if $a<0, b=0$, and upper half plane if $a=0, b>0$.

b) The properties formalized in Definitions \ref{defSINHLevy} and \ref{defSINHLevyell}
for characteristic exponents hold for wide classes of the symbols of pseudo-differential operators (PDO)
which have no relation to probability, and the methods based on these properties can be applied to develop efficient numerical methods for
various boundary problems for such PDOs.
}
\end{rem}

\begin{rem}\label{rem_ell_str-subord}
{\rm If a subordinator $Y$ and process $X$ are elliptic SINH-regular L\'evy processes, then $\{X_{Y_t}\}$ is
an  elliptic SINH-regular L\'evy process. See Section \ref{subord} for an example.
}
\end{rem}
\subsection{Examples and some generalizations}\label{genexSINH}
\begin{enumerate}[(1)]
\item
 Essentially all L\'evy processes
used in quantitative finance are elliptic SINH-regular L\'evy processes:
 Brownian motion (BM), Merton model \cite{merton-model}, NIG (normal inverse Gaussian model) \cite{B-N}, hyperbolic
 processes \cite{EK}, double-exponential jump-diffusion model \cite{lipton-risk,lipton-columbia,kou,KW1,KW2},
 its generalization: hyper-exponential jump-diffusion model,
  introduced in \cite{ amer-put-levy-maphysto,lipton-risk} and studied in detail in  \cite{amer-put-levy-maphysto, amer-put-levy},
 the majority of processes of the $\beta$-class \cite{beta}; the generalized
 Koponen's family \cite{KoBoL} and its subclass KoBoL
 \cite{NG-MBS}. A subclass of KoBoL
(known as the CGMY model - see \cite{CGMY}) is given by the characteristic exponent
\begin{equation}\label{kbl}
\psi(\xi)=-i\mu\xi+c\Gamma(-\nu)[\lp^\nu-(\lp+i\xi)^\nu+(-\lm)^\nu-(-\lm-i\xi)^\nu],
\end{equation}
where $\nu\in (0,2), \nu\neq 1$ (in the case $\nu=1$, the analytical expression is different: see \cite{KoBoL,NG-MBS}).
Thus, KoBoL is SINH-regular of type $((\lm,\lp); \cC, \cC_+)$ and order $\nu$, where $\cC=\cC_{-\ga,\ga}$, $\ga\ge \pi/2$, and
$\cC_+=\cC_{-\ga', \ga'}$, where  $\ga'=\pi/(2\nu)$ (see (3) below for the meaning of $\ga'>\pi/2$).
 BM, DEJD and HEJD are of order $\nu=2$, and NIG is of order $\nu=1$.

 The characteristic exponents of NTS processes  constructed in \cite{B-N-L} are given by
 \begin{equation}\label{NTS2}
\psi(\xi)=-i\mu\xi+\de[(\al^2+(\xi+i\be)^2)^{\nu/2}-(\al^2-\be^2)^{\nu/2}],
\end{equation}
where $\nu\in (0,2)$, $\de>0$, $|\be|<\al$. This is a process of type $((\al+\be,\al-\be); \cC, \cC_+)$ of order $\nu$,
where $\cC$ and $\cC_+$ are the same as for KoBoL of the same order.

\item
 In order to consider Variance Gamma processes (VG) \cite{MM91}, Definitions \ref{defSINHLevy}-\ref{defSINHLevyell} must be generalized
replacing the function $\rho\mapsto \rho^\nu$ with a strictly increasing function
$w:\bR_{+}\to\bR_{+}$ satisfying $w(+\infty)=+\infty$. We say: $X$ is an (elliptic) SINH-regular L\'evy process of type
$(S,\cC,\cC_+, w)$. For Variance Gamma processes, $w(\rho)=\ln(1+\rho)$.

\item
  For KoBoL, VG  and NTS, $\psi^0$ admits analytic
continuation to an appropriate Riemann surface $\cR$, and $\cC$ can be defined as an appropriate subset of $\cR$.
Formally, in (1), $\ga>\pi/2$ is admissible with the understanding that seemingly overlapping parts of $\cC_{-\ga,\ga}$ lie on
different sheets of $\cR$.
See \cite{iFT,paraLaplace,paired},
where advantages of $\cC_+\subset \cR$ were utilized to increase the speed. The same extension is very useful
when the SINH-acceleration is applied to calculate the Wiener-Hopf factors, and less so for pricing European options.

\item
The asymptotic coefficient $c_\infty(\mathrm{arg}\,\xi_0)$ is
 \begin{enumerate}[(i)]
 \item
 if $X$ is BM, DEJD and HEJD, $c_\infty(\varphi)=(\sg^2/2)e^{i2\varphi}$, hence, $\cC_+=\cC_{-\pi/4,\pi/4}$;
 \item
  if $X$ is given by \eqref{kbl}, then, for $\varphi\in [-\pi/(2\nu),\pi/(2\nu)]$,
 \bbe\label{asKoBoL}
 c_\infty(\varphi)=-2c\Ga(-\nu)\cos(\pi\nu/2)e^{i\nu\varphi},
 \ee
 hence, $\cC_+=\cC_{-\pi/(2\nu),\pi/(2\nu)}$;
 \item
  if $X$ is given by \eqref{NTS2}, then, for $\varphi\in [-\pi/(2\nu),\pi/(2\nu)]$,
  \bbe\label{asNTS}
 c_\infty(i\varphi)=\de e^{i\nu\varphi},
 \ee
 hence, $\cC_+=\cC_{-\pi/(2\nu),\pi/(2\nu)}$.
\end{enumerate}
\item In \cite{KoBoL}, we constructed more general classes of L\'evy processes, with the characteristic exponents of the form
\begin{equation}\label{kblgen}
\psi(\xi)=-i\mu\xi+c_+\Gamma(-\nu_+)[\lp^{\nu_+}-(\lp+i\xi)^{\nu_+}]+c_-\Gamma(-\nu_-)[(-\lm)^{\nu_-}-(-\lm-i\xi)^{\nu_-}],
\end{equation}
where $c_\pm\ge 0$, $c_++c_->0$, $\lm<0<\lp, \nu_\pm\in (0,2), \nu_\pm\neq 1$, with  modifications in the case $\nu_+=1$ and/or $\nu_-=1$.
For these processes, the domains of analyticity and bounds are more involved. In particular, in general, the coni are not symmetric w.r.t. the real axis.
\item
In examples above, $\Re\psi^0(\xi)\to +\infty$ as $\xi\to\infty$ in a conus around the real axis, due to special forms
of the L\'evy densities, hence, characteristic functions. If $\psi^0(\xi)$ contains terms of the form $Ce^{ic\xi}$, where $c\in\bR$,
then
\begin{enumerate}[(i)]
\item
if $c<0$, then $\Re\psi^0(\xi)$ is not semibounded (from below) as $\xi\to\infty$ in any conus $\cC$ in the upper half-plane;
\item
 if $c>0$,
then $\Re\psi^0(\xi)$ is not semibounded  (from below)  as $\xi\to\infty$ in any conus in the lower half-plane.
\end{enumerate}
The simplest example is the BM with the embedded jumps, the L\'evy density being $\bfo_{[-a,b]}$. If $a<0<b$, then
$\Re\psi^0(\xi)$ is not semibounded as $\xi\to\infty$ in any conus $\cC$.
\item
If $X$ is the BM with embedded negative jumps only, and the jump density decays exponentially at infinity, then $X$ is an elliptic SINH-regular process of order 2 and type $((0,+\infty);  \cC, \cC_{0,\pi/4})$, where
$\cC$ is the upper half-plane. If $X$ is the BM motion with embedded positive jumps only, then $X$ is 
an elliptic SINH-regular process of order 2 and type $((-\infty,0);
\cC, \cC_{-\pi/4,0})$, where
$\cC$ is the lower half-plane.
\item
In Example (7), one may any add a positive (resp., negative) jump component as in KoBoL or exponential jump-diffusions. One can also
replace the BM with one-sided SINH-regular processes of type $((\mum,\mup), \cC, \cC_+)$, where  $\mum<0<\mup$ and $\cC_+, \cC$ are coni around the real axis.
In both cases, the type of the resulting process $X$ will be characterized by smaller domains of analyticity than in the case
of the BM with embedded one-sided jump components.
\item
Conditional distributions in affine stochastic volatility models and affine and quadratic interest rate models are sinh-regular.
\end{enumerate}

\subsection{Calculation of probability distributions}\label{SINHpdf}
The pdf of $X_t$ equals
\begin{equation}\label{pdfLevy}
p_t(x)=\frac{1}{2\pi}\int_\bR e^{-ix'\xi-t\psi^0(\xi)}d\xi,
\end{equation}
where $x'=x-\mu t$. Denote $g(\xi)=e^{-t\psi^0(\xi)}$. The change of variables \eqref{sinhbasic}
can be justified if the integrand $f(y)=e^{-ix' \chi_{\om_1,\om; b}(y)}g(\chi_{\om_1,\om; b}(y))\chi_{\om_1,\om; b}'(y)$ admits analytic continuation
to a strip $S_{(-d,d)}=\{y\in \bC\ |\ \Im y\in (-d,d)\}$ around the real line and decays sufficiently fast as $y\to\infty$ remaining
in the strip. In more detail, the Cauchy integral theorem allows us to deform the line of integration
$\{\Im\xi=\om_0\}$ into the contour $\cL_{\om_1,\om; b}:=\chi_{\om_1, \om; b}(\bR)$. In the integral over
$\cL_{\om_1, \om; b}$, we make the change of variables \eqref{sinhbasic}.
\begin{figure}
\scalebox{0.9}
{\includegraphics{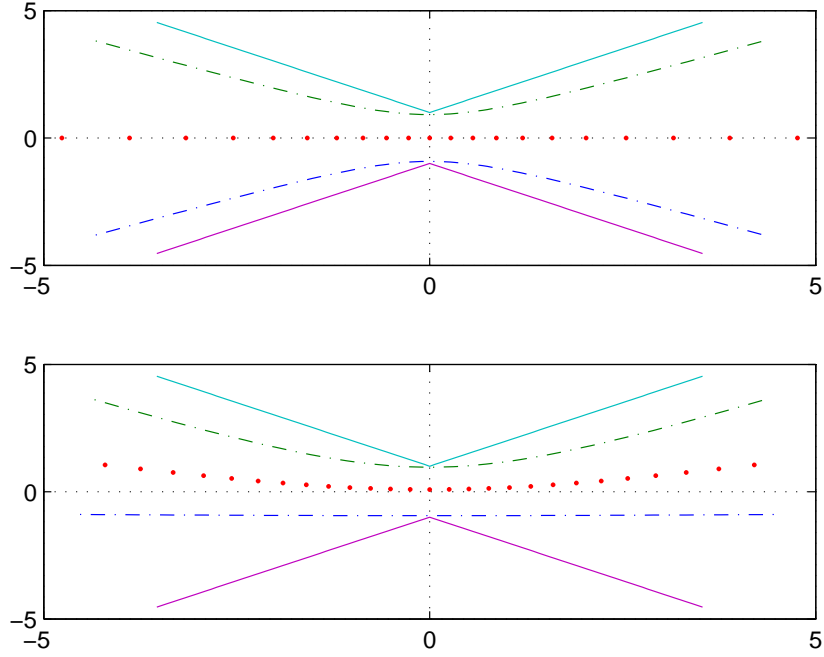}}
\caption{\small Solid lines: boundaries of the domain of analyticity $S_{(-1,1)}+\cC_{-\pi/4,\pi/4}$
in $\xi$-coordinate.
Dots: points $\xi_j=\chi_{\om_1,\om;b}(y_j)=i\om_1+b\sinh(i\om +y_j)$ used in the simplified trapezoid rule.
Dots-dashes: boundaries of the image $\chi_{\om_1,\om;b}(S_{(-d,d)})$ of the strip of analyticity $S_{(-d,d)}$.
 Upper panel: $\om_1=\om=0$, $d=\pi/4$, $b=1/\sin(\pi/4)$. Lower panel: $\om_1=-1, \om=d=\pi/8$, $b=2/\sin(\pi/8)$.
 For the calculations represented in the lower panel, only a smaller domain $S_{(-1,1)}+\cC_{0,\pi/4}$ matters.}

\end{figure}

The choice of the parameters of the sinh-acceleration depends on the type of the process,
its order, and $x'$. It is convenient to consider separately the following cases:
\begin{enumerate}[(1)]
\item $\cC_+=\cC_{\gam,\gap}$, where $\gam<0<\gap$,
$\nu\in (1,2]$ or $\nu\in (0,1]$ and $x'=0$: as
$\xi\to\infty$, the leading term of the asymptotics of $-ix'\xi-t\psi^0(\xi)$ is the same as the one of $-t\psi^0(\xi)$,
hence, one can choose the parameters of the sinh-acceleration without taking $x'$ into account.
The line of integration may remain flat or it can be deformed either upward or downward. The conus of analyticity
used to derive the error bound can be around the real axis (see the upper panel in Fig.1 for illustration), which allows one to use the mesh of a larger size than
in the other cases.  Naturally, if
$\nu>1$ is close to 1 and/or $x'$ is large in the absolute value, then it is safer to take the sign of $x'$ into account,
and deform the contour as in the case $\nu\in (0,1)$.

\item
$\cC_+=\cC_{\gam,\gap}$, where $\gam=0<\gap$,
$\nu\in (0,1)$, $x'<0$. As
$\xi\to\infty$, the leading term of the asymptotics of $-ix'\xi-t\psi^0(\xi)$
 is the same as the one of  $-ix'\xi$, hence, the deformed contour and the conus of analyticity
used to derive the error bound must be in the upper half-plane even if $\gam<0$ (see the lower panel in Fig.1 for illustration).
\item
$\cC_+=\cC_{\gam,\gap}$, where $\gam<0=\gap$,
$\nu\in (0,1)$, $x'>0$. As
$\xi\to\infty$, the leading term of the asymptotics of $-ix'\xi-t\psi^0(\xi)$  is the same as the one of  $-ix'\xi$, hence, the deformed contour and the conus of analyticity
used to derive the error bound must be in the lower half-plane even if $\gap>0$.

\item $\cC_+=\cC_{\gam,\gap}$, where $\gam<0<\gap$,
$\nu=1$, $x'\neq 0$. It is optimal to deform the contour but the conus of analyticity
used to derive the error bound can be around the real axis.
\item $\cC_+=\cC_{\gam,\gap}$, where either $\gam=0<\gap$ or $\gam<0=\gap$, and
$\nu=1$. The deformed contour and the conus of analyticity
used to derive the error bound must be in the lower half-plane even if $\gap>0$ and in the lower half-plane even if $\gam<0$.
\end{enumerate}
In the next subsection, we assume that $\cC_+=\cC_{\gam,\gap}$, where $\gam<0<\gap$. The reader can easily
modify the constructions below for the cases $\gam=0<\gap$ and $\gam<0=\gap$.

\subsubsection{The case $\nu\in (1,2]$ and the case
$\nu\in (0,1]$, $x'=0$}\label{pdfnuge1}
In these cases, for any $\ga^{-,'}\in (\gam,0), \ga^{+,'}\in (0,\gap)$,   there exists $c>0$ such that
\bbe\label{cinf1prime}
\Re c_\infty(\varphi)\ge c, \quad \ga^{-,'}\le \varphi\le \ga^{+,'}.
\ee
First, we choose $\om\in \bR$ and $d_0>0$ so that $\om+d_0\le \gap, \om-d_0\ge \gam$. Since $d_0$ is the upper bound
for the half-width of the strip of analyticity in the $y$-coordinate, we want to choose $d_0$ as large as possible. Hence, we set
\bbe\label{omd0}
\om=(\gap+\gam)/2,\ d_0=(\gap-\gam)/2;
\ee
then $\om+d_0=\gap, \om-d_0=\gam$.
Next, we must ensure that the intersection of the imaginary axis and the image of $S_{(-d_0,d_0)}$ under $\chi_{\om_1, \om; b}$
is a subset of $(\mum,\mup)$, which is equivalent to
$
\om_1+ba_+\le \mup, \om_1-b \mum\ge a_-,$ where
$a_-= \sin (\min\{\pi/2,-\gam\}), a_+=\sin (\min\{\pi/2,\gap\})$.
We define
\bbe\label{om1b0}
\om_1=\frac{\mup a_-+\mum a_+}{a_++a_-},\
b_0=\frac{\mup-\mum}{a_++a_-}.
\ee
If $\gam=-\gap$ (the case for KoBoL and NTS, the generalization of NIG), \eqref{om1b0} simplify
\bbe\label{om1b0s}
\om_1=\frac{\mup+\mum}{2},\
b_0=\frac{\mup-\mum}{2\sin (\min\{\pi/2,\gap\})}.
\ee
In \eqref{sinhbasic}, we choose $d<d_0$, $b<b_0$ close to $d_0, b_0$, respectively,
e.g., $d=0.95 d_0, b=0.95d_0$.
Then the integrand
\[
f(y)=(1/2\pi)\exp[-ix' \chi_{\om_1,\om; b}(y)-t\psi^0(\chi_{\om_1,\om; b}(y))]\chi_{\om_1,\om; b}'(y)
\] admits analytic continuation
to the strip $S_{(-d,d)}=\{y\in \bC\ |\ \Im y\in (-d,d)\}$ around the real line and decays sufficiently fast as $y\to\infty$ remaining
in the strip so that
\[
\lim_{A\to \pm\infty}\int_{-d}^d |f(ia+A)|da=0,\]
and the Hardy norm
\bbe\label{Hnorm}
H(f,d)=\lim_{a\downarrow -d}\int_\bR|f(ia+ y)|dy+\lim_{a\uparrow d}\int_\bR|f(ia+y)|dy
\ee
is finite.
Fix $\ze>0$ and
construct the grid $\{y_j=j\ze, j\in \bZ\}$. The discretization error
of the infinite trapezoid rule
\begin{equation}\label{pdfLevytrap}
p_t(x)=\ze\sum_{j\in \bZ} f(y_j)
\end{equation}
 admits an upper bound via $H(f,d)\exp[-2\pi d/\ze]/(1-\exp[-2\pi d/\ze])$  (see Theorem 3.2.1 in \cite{stenger-book} and Appendix in \cite{paraHeston} for a simple proof).
 In some cases, the Hardy norm of the integrand as a function on a maximal strip of analyticity
is infinite. In such cases, in order to use the universal bound for the discretization error, one has to apply the bound
to functions on
a narrower strip of analyticity; this explains our choices  $d<d_0$ and $b<b_0$. A fairly accurate approximate bound
for $H(f,d)$ can be derived relatively easily but, as in the case of the fractional-parabolic deformations \cite{paraHeston,pitfalls}, the following crude approximation
works well if the initial strip of analyticity is not very narrow and, typically, leads to an overkill:
\bbe\label{hfdbound}
H(f,d)=C(|f(-id)|+|f(id)|),
\ee
where $C=10$.
 To satisfy a small error tolerance $\eps>0$,
we choose $\ze =2\pi d/(\ln(H(f,d)/\eps)\sim 2\pi d/E$, where $E=\ln(1/\eps)$.
The choice of $N$, the number of terms of the simplified trapezoid rule, equivalently,
of the truncation parameter $\La=N\ze$, is somewhat more involved. The truncation error of the simplified trapezoid rule
\begin{equation}\label{pdfLevysimpltrap}
p_t(x)=\ze\sum_{|j|\le N} f(y_j)
\end{equation}
can be approximated by the truncation error of the integral
\begin{equation}\label{truncerror}
Err_{tr}=2\int_\La^{+\infty} |f(y)|dy.
\end{equation}
If $\La$ is large, then, on $[\La,+\infty)$, we can use an approximation
\[
\chi_{\om_1,\om; b}(y)\sim (b/2)e^y e^{i\om},\ \chi_{\om_1,\om; b}'(y)\sim (b/2)e^y e^{i\om},\]
to derive an approximate upper bound
\[
2\int_\La^{+\infty} |f(y)|dy\le \frac{1}{\pi}\int_{\La_1}^{+\infty}e^{\Re(-ix'e^{i\om}\rho-t\psi(\rho e^{i\varphi})}d\rho,
\]
where $\La_1=(b/2)e^\La$. The bound can be simplified (at the expense of some loss in accuracy) as
\bbe\label{truncerror2}
Err_{tr}(\La_1)\le \frac{e^{tC_0}}{\pi}\int_{\La_1}^{+\infty} e^{(x'\sin\om)\rho-t\Re c_\infty(\om)\rho^\nu}d\rho,
\ee
where $C_0=c\Ga(-\nu)[\lp^\nu+(-\lm)^\nu]$ in the case of KoBoL, and $C_0=\de(\al^2-\be^2)^{\nu/2}$ in the case of NTS processes.
If $c_\infty(\om)=c_\infty(0)e^{i\om\nu}$ as in the cases of KoBoL and NTS, then \eqref{truncerror2} can be written as
\bbe\label{truncerror3}
Err_{tr}(\La_1)\le \frac{e^{tC_0}}{\pi}\int_{\La_1}^{+\infty} e^{(x'\sin\om)\rho-t c_\infty(0)\cos(\om\nu)\rho^\nu}d\rho.
\ee
Given the error tolerance $\eps>0$, one can find an approximation to $\La_1$ satisfying $Err_{tr}(\La_1)<\eps$
quite easily (see \cite{iFT, paraHeston, pitfalls}), and then define
\bbe\label{LaN}
\La=\ln(2\La_1/b),\ N=\mathrm{ceil}(\La/\ze).
\ee

\subsection{The case $\nu\in (0,1)$, $x'<0$} Since $-ix'\xi-t\psi^0(\xi)\sim -ix'\xi$ as $\xi\to\infty$, we use the same construction as above replacing $\gam$ and $\gap$ with
$\gam_0=0$ and $\gap_0=\min\{\gap,\pi\}$.

\subsection{The case $\nu\in (0,1)$, $x'>0$} We use the same construction as above replacing $\gam$ and $\gap$ with
$\gam_0=\max\{\gam,-\pi\}$ and $\gap_0=0$.

\subsection{The case $\nu=1$, $x'\neq 0$} For simplicity, consider the case $c_\infty(\varphi)=c_\infty(0) e^{i\varphi}$, where $c_\infty(0)>0$ is
independent of $\varphi$.  As $\rho\to+\infty$,
\[
-ix'\rho e^{i\varphi}-t\psi^0(\rho e^{i\varphi})=(-ix'-tc_\infty(0))e^{i\varphi}=-\sqrt{(x')^2+(tc_\infty(0))^2}e^{i(\varphi+\varphi_0)},
\]
where $\varphi_0=\arctan(x'/(tc_\infty))$. Hence, we use the same constructions as above with
\begin{enumerate}[(i)]
\item
$\gam=-\pi/2-\varphi_0$, $\gap=\pi/2$, if $x'<0$ (hence, $\varphi_0<0$);
\item
$\gam=-\pi/2, \gap=\pi/2-\varphi_0$, if $x'>0$ (hence, $\varphi_0>0$).
\end{enumerate}
The bound for the truncation error \eqref{truncerror2} can be made explicit
\bbe\label{truncerror22}
Err_{tr}(\La_1)\le \frac{e^{tC_0}}{\pi(x'\sin\om+t c_\infty\cos\om)}e^{-(-x'\sin\om+t c_\infty(0)\cos\om)\La_1}.
\ee
Given $\eps>0$, we find
\bbe\label{La1nu1}
\La_1=\frac{\ln(1/\eps)+tC_0-\ln(\pi(-x'\sin\om+t c_\infty\cos\om))}{x'\sin\om+t c_\infty\cos\om},
\ee
and then apply \eqref{LaN}.

\subsection{Complexity of the scheme based on the sinh-acceleration} As $\eps\downarrow 0$, $\La\sim \ln E,$ where
$E=E(\eps)=\ln(1/\eps)$, and
$\ze\sim E/(2\pi d)$, where $d<(\gap-\gam)/2$ is fixed. Hence, the complexity of the scheme is of the order of
$A(d)E\ln E$, where $A(d)$ can be arbitrarily close to $1/(\pi (\gap-\gam))$ if $d$ is chosen sufficiently close to $(\gap-\gam)/2$.
Note that if the integrand admits the analytic continuation to an appropriate Riemann surface, then $\gap$ and/or $-\gam$ can be larger than $\pi/2$.
In particular,  then, for NTS and KoBoL of order $\nu\in (0,1)$,
\begin{enumerate}[(i)]
\item
if $x'=0$, then $\gap=-\gam=\pi/(2\nu)$,
\item if $x'<0$, then $\gap=\min\{\pi, \pi/(2\nu)\}$, $\gam=0$;
\item
if $x'>0$, then $\gap=0$,  $\gam=-\min\{\pi, \pi/(2\nu)\}$.
\end{enumerate}
This implies that, rather counter-intuitively, the (asymptotic) complexity of the scheme decreases with $\nu$, whereas for
the flat iFT and the scheme based on the fractional-parabolic deformations, the complexity of the scheme increases as $\nu$ decreases.

\begin{rem}\label{b-rem} {\rm The approximate bound for the complexity of the scheme derived above implicitly assumes that the strip is neither too wide
nor too narrow, hence, $b$ is neither too large nor too small. If the width of the strip becomes too large or too small,
the Hardy norm may become too large. Hence, the approximations $\ln(2\La_1/b)\sim \ln \La_1$
and $\ln(H/\ze)\sim \ln(1/\eps)$
which we used to access the complexity of the scheme may become not very accurate. The problem of a too wide strip can be fixed  using a moderately wide strip instead of a very wide one.

If the strip $S_{(\mum,\mup)}$ is too narrow, but the Hardy norm does not tend to infinity as the strip shrinks (the case of the NTS model and KoBoL models),
then the sinh-change of variables implies the rescaling which reduces the calculations to the case of a strip $S_{(-d,d)}$, where $d=k_dd_0$, $d_0=(\gap-\gam)/2$,
$k_d=0.9-0.95$. Hence,
a more accurate approximation of the truncation parameter is $\ln\La_1+\ln(1/(\mup-\mum))$ instead of $\ln\La_1$. Even if the initial strip is extremely narrow, say, of the width
$10^{-8}$, the recommended truncation parameter increases by less than 30, and the number of terms needed to satisfy even vary small error tolerance
remains quite moderate. In the case of the fractional-parabolic change of variables
of order $\al$ (typically, $\al\in (1, 2)$),
the truncation parameter increases by a factor of the order of $(\mup-\mum)^{-1/\al}$, which can be large if $\mup-\mum$, the
width of the strip, is very small (see Section \ref{complexityFrPar} for the choice of the parameters
of the fractional-parabolic method and analysis of its complexity). However, it is important that the angle $\gam-\gap$ between the rays $e^{i\gam}\bR_+$ and $e^{i\gap}\bR_+$ be not too small.
}
\end{rem}

\subsection{Numerical example}
In Section \ref{tablespdfNTS}, Tables 1 and 2, we show the pdf of $X_t$ 
with the characteristic exponent \eqref{NTS2} calculated using the sinh-acceleration, fractional parabolic change of variables and the standard inverse Fourier transform method
(flat iFT). The parameters of the process are  $\mu=0$, $\al=10, \be=0$
for $t=0.004$; $\de=m_2\la^{\nu-2}$, where $m_2=\psi^{\prime\prime}(0)=0.1$ is the second instantaneous moment.
In Table 1, $\nu$ varies, and the pdf is calculated at the peak. In Table 2, $\nu=0.3$ is fixed,
and $x$ varies.

For the flat iFT, we use the accurate prescriptions for the choice of the mesh size $\ze$ derived in \cite{iFT}, for the error tolerance
$\eps=10^{-7}$. By trial and error, we find that, due to the oscillation of terms in
the infinite trapezoid rule, it is possible to increase the size of the mesh: $\ze/k_\ze$ where $k_\ze=0.6$. Hence, we may use a smaller number of terms in the simplified trapezoid rule
to satisfy a given error tolerance for the truncation error. Nevertheless, as the results shown in Tables 1 and 2 demonstrate,
the flat iFT may require extremely large number of terms. At the same time, the sinh-acceleration allows one  to satisfy a small error tolerance
with several dozen of terms; fractional-parabolic method is less efficient than the sinh-acceleration.

 It is easily see that unless $\nu$ is not small, equivalently, the process is close to the BM,
  it is essentially impossible to calculate the pdf at the peak sufficiently accurately and fast
 which is needed for an efficient MLE. A similar problem arises when the cumulative pdf is calculated and applied for
 simulation purposes.

\section{Meromorphic SINH-regular  distributions on $\bR$}\label{simpleII}
\subsection{Cumulative pdf in SINH-regular models} In this case, we need to evaluate the integral
of the same type but with an additional factor $-1/(i\xi)$ under the integral sign:
\bbe\label{cpdf0}
\bP[X_t<x]=\frac{1}{2\pi}\int_{\Im\xi=\om_0}\frac{e^{-ix'\xi-t\psi^0(\xi)}}{-i\xi}d\xi,
\ee
where $\om_0\in (0,\mup)$. Since there is a pole at 0, we can apply the same scheme as in Section \ref{SINHpdf}
replacing $\mum$ with 0; the truncation parameter will be somewhat smaller due to the additional decaying factor
$1/(i\xi)$. To be more specific, instead of \eqref{truncerror2}, we have the error bound
\bbe\label{truncerror2cpdf}
Err_{tr}(\La_1)\le \frac{e^{tC_0}}{\pi}\int_{\La_1}^{+\infty} e^{-(x'\sin\om)\rho-t\Re c_\infty(\om)\rho^\nu}\rho^{-1}d\rho.
\ee
If $x'>0$ or $\mup$ is small and $-\mum>\mup$, it is advantageous to move the line of integration
in \eqref{cpdf0} down, and, on crossing the simple pole, apply the residue theorem:
\bbe\label{cpdf2}
\bP[X_t<x]=1+\frac{1}{2\pi}\int_{\Im\xi=\om'_0}\frac{e^{-ix'\xi-t\psi^0(\xi)}}{-i\xi}d\xi,
\ee
where $\om'_0\in (\mum,0)$. The integral on the RHS of \eqref{cpdf2} is calculated as in Section \ref{SINHpdf}, with
$\mup=0$.

\subsection{Puts and calls in SINH-regular L\'evy models}
Let $r$ be the riskless rate, $\tau$ the time to maturity, $K$ the strike, and $S$ the spot.
Set $x'=\ln(S/K)+\mu \tau$. Assuming that $\mum<-1$, the price of the call option is given by
\bbe\label{call}
V_{\mathrm{call}}(K;\tau, S)=-\frac{Ke^{-r\tau}}{2\pi}\int_{\Im\xi=\om_0}\frac{e^{ix'\xi-\tau\psi^0(\xi)}}{\xi(\xi+i)}d\xi,
\ee
where $\om_0\in (\mum, -1)$. The put price is given by the same integral but with
$\om_0\in (0,\mup)$, and the price of the covered call by the same integral but with  $\om_0\in (-1,0)$.

For the call, we use the same scheme as above with $\mup$ replaced with  $-1$, for the put, $\mum$ is replaced by $0$, and
for the covered call, we use $\mum=-1, \mup=0$. If $x'=0$, we use the $\gam<0<\gap$ from the definition of the sinh-regular process
to define $\om=(\gap+\gam)/2$ and use the conus $\cC_{\gam,\gap}$ to derive the recommendation for the choice of $\ze$ and $N$;
if $x'>0$, we replace $\gam$ with 0 so that the wings of the deformed contour point upward and the factor $e^{ix'\xi}$ decays as $\xi\to \infty$
in the conus $\cC_{0,\gap}$ used to derive the recommendations for the choice of $\ze$ and $N$, and set $\om=\gap/2$;
if $x'<0$, we replace $\gap$ with 0 so that the wings of the deformed contour point downward
and the factor $e^{ix'\xi}$ decays as $\xi\to \infty$
in the conus $\cC_{\gam,0}$ used to derive the recommendations for the choice of $\ze$ and $N$, and set $\om=\gam/2$.
Note that in all cases, $A(\om):=-x'\sin\om\ge 0, B(\om)=\tau\Re c_\infty(\om)>0$.

The bound for the truncation error is
\bbe\label{truncerror2call}
Err_{tr}(\La_1)\le \frac{e^{\tau C_0}}{\pi}\int_{\La_1}^{+\infty} e^{(x'\sin\om)\rho-\tau\Re c_\infty(\om)\rho^\nu}\rho^{-2}d\rho.
\ee
We can use a more stringent bound
\bbe\label{truncerror2call2}
Err_{tr}(\La_1)\le \frac{e^{\tau C_0}}{\pi \La_1} e^{(x'\sin\om)\La_1-\tau\Re c_\infty(\om)\La_1^\nu}.
\ee
Given the error tolerance $\eps>0$, we have the equation for $\La_1$:
\bbe\label{eqLa1}
F(\La_1):=A(\om)\La_1+B(\om)\La_1^\nu+\ln\La_1-C=0,
\ee
where $C=\tau C_0-\ln(\pi \eps)$. The equation can be
solved easily and fast since it is unnecessary to achieve a high precision.

\subsection{Pricing European puts and calls in the Heston model}\label{EuroHeston}
Consider the Heston model \cite{heston-model} with constant riskless and dividend rates $r$ and $\de$
on stock (or exchange rate) $S_t$. To be more specific, we assume that, under an EMM $\bQ$ chosen for pricing,
  $S_t$ and the stock volatility $v_t$
follow the system of stochastic differential equations (SDE)
\begin{eqnarray}\label{sdeS}
 \frac{dS_t}{S_t}& =&(r-\de) dt + \sqrt{v_t}d\hat W_{1,t},\\\label{sdev}
  d v_t& = &\ka(m-v_t)dt +
  \sg_0\sqrt{v_t}dW_{2,t},
  \end{eqnarray}
  where $\hat W_{1,t}, W_{2,t}$ are components of the Brownian motion in 2D with unit variances
  and correlation coefficient $\rho$. Starting with \cite{heston-model}, prices of European options in the Heston
  model have been calculated using the Fourier transform technique (the first instance of using this standard technique
  in finance). For an overview of different realizations of the pricing formula, see \cite{paraHeston};
  we use the realization derived in \cite{paraHeston}.  Let $V_{\rm put}(T,K; t, S_t, v_t)$ and $V_{\rm call}(T,K; t, S_t, v_t)$
  be the put and call options on $S_t$ with strike $K$ and maturity date $T$, at time $t<T$.
\begin{thm}[\cite{paraHeston}]\label{thmEuroprice}  Let $\tau=T-t$ be the time to maturity and let
$\lm(\tau)<-1<0<\lp(\tau)$ be reals such that
  $\bE^\bQ[S_T^{\la_\pm(\tau)}\ |\ S_t,  v_t]<\infty$. Then, for any $\om_0 \in (0,\lp(\tau))$,
  \begin{equation}\label{putprice}
   V_{\rm put}(T, K; t, S_t, v_t)=-\frac{Ke^{-r\tau}}{2\pi}
   \int_{\Im\xi=\om_0}\frac{e^{i\xi z_t+(v_tB_0(\tau,\xi)+C_0(\tau,\xi))/\sg_0^{2}}}{\xi(\xi+i)}d\xi,
\end{equation}
and for any $\om_0\in (\lm(\tau),-1)$,
\begin{equation}\label{callprice}
   V_{\rm call}(T,K; t, S_t, v_t)=-\frac{Ke^{-r\tau}}{2\pi}
  \int_{\Im\xi=\om_0}\frac{e^{i\xi z_t+(v_tB_0(\tau,\xi)+C_0(\tau,\xi))/\sg_0^{2}}}{\xi(\xi+i)}d\xi,
\end{equation}
where $z_t=\log(S_t/K)-(\rho/\sg_0)v_t+\mu_0\tau$,
$\mu_0=r-\de-\ka m\rho/ \sg_0$,
\begin{eqnarray}\label{defB0}
B_0(\tau,\xi)&=& (\ka-R(\xi))\frac{1-D_1(\xi)e^{-\tau R(\xi)}}{1-D(\xi)e^{-\tau R(\xi)}}\\\label{defC0}
C_0(\tau,\xi)&=&\ka m \left((\ka-R(\xi))\tau-2\ln\frac{1-D(\xi)e^{-\tau R(\xi)}}{1-D(\xi)}\right),
\\\label{defR0}
R(\xi)&=&\sqrt{\ka^2+(\sg_0^2-2\rho\ka)i\xi+\sg_0^2(1-\rho^2)\xi^2}\\\label{defD0}
D(\xi)&=&\frac{\rho \sg_0i\xi-\ka+R(\xi)}{\rho \sg_0i\xi-\ka-R(\xi)}\\\label{defD01}
D_1(\xi)&=&D(\xi)\frac{\ka+R(\xi)}{\ka-R(\xi)}
\end{eqnarray}
\end{thm}
\begin{rem}\label{rem:Heston1}{\rm
a) Let $\{\Im\xi\in (\lm(\tau), \lp(\tau))\}$ be the maximal strip of analyticity of the characteristic function. Introduce a quadratic polynomial
$
P(\be)=\ka^2-(\sg_0^2-2\rho\ka)\be-\sg_0^2(1-\rho^2)\be^2,$
and denote by $\lm^0<0<\lp^0$ its roots. It is proved in \cite{paraHeston}
that $\lm(\tau)\le \lm^0$ and $\lp^0\le \lp(\tau)$, and a procedure for
the calculation of $\lm(\tau)$ and $\lp(\tau)$ is derived. As numerical examples in \cite{paraHeston} indicate,
typically, $\lm^0$ and $\lp^0$ are rather close to $\lm(\tau)$ and $\lp(\tau)$, respectively, hence,
there is no sizable advantage in using a rather complicated procedure for the calculation of $\lm(\tau)$ and $\lp(\tau)$.
In the numerical procedure of the present paper, we use $\lm^0$ and $\lp^0$.

b) As it is proved in \cite{Lucic}, the (conditional) characteristic function admits the analytic continuation to
the complex plane with the cuts $i(-\infty,\lm(\tau)]$ and $i[\lp(\tau),+\infty)$, hence, to the complex plane with the cuts $i(-\infty,\lm^0]$ and $i[\lp^0,+\infty)$.

}
\end{rem}
Choose the strip of analyticity $S_{(\lm^0,-1)}$, $S_{(-1,0)}$, or $S_{(0,\lp^0)}$,
and move the line of integration into the strip; use
the residue theorem when a pole (or the two poles) of the integrand is (are) crossed.
Unless a strip is too narrow or wide (see Remark \ref{b-rem}), one should choose a curve in the upper half-plane if $z_t>0$, and in the lower
half-plane if $z_t<0$. Otherwise, the universal recommendation for the choice of the truncation parameter
becomes inaccurate: one must add the length of the part of the deformed contour which is in the ``incorrect" half-plane
where the factor $e^{iz_t\xi}$ is large in the absolute value.
\sbr
 Let $S_{(\mum,\mup)}$ be the chosen strip.
It follows from Remark \ref{rem:Heston1} that the conditional distribution of the price is sinh-regular of order $\nu=1$ and type
$((\mum,\mup);  \cC_{-\pi/2,\pi/2}, \cC_{-\gam,\gap})$,
where $\ga^\pm$ are defined as for elliptic L\'evy processes of order $\nu=1$ with $x'=z_t$ and
\[
c_\infty(\varphi)=\frac{1}{\sg_0^2}\lim_{\rho\to+\infty}\rho^{-1}(v_tB_0(\tau,\rho e^{i\varphi})+C_0(\tau,e^{i\varphi})).\]
To find $c_\infty(\varphi)$, we calculate the asymptotics of $R(\xi), D(\xi), B_0(\tau,\xi)$ and $C_0(\tau,\xi)$ as $\xi\to\infty$ remaining in the right-half plane:
\beqa\label{Rxias}
R(\xi)&=&\sg_0(1-\rho^2)^{1/2}\xi\left(1+\frac{\sg_0^2-2\rho\ka}{\sg_0^2(1-\rho^2)}i\xi^{-1}+O(\xi^{-2})\right)^{1/2}\\\nonumber
&=&\sg_0(1-\rho^2)^{1/2}\xi+ i\frac{\sg_0^2-2\rho\ka}{2\sg_0^2(1-\rho^2)}+O(\xi^{-1})\\\label{1Dxi}
1-D(\xi)&=&\frac{-2R(\xi)}{\rho\sg_0i\xi-\ka-R(\xi)}
=\frac{-2}{\rho\sg_0i\xi/R(\xi)-1+O(\xi^{-1})}\\\nonumber
&=&\frac{2}{1-i\rho/(1-\rho^2)^{1/2}}+O(\xi^{-1})\\\label{B0C0}
v_tB_0(\tau,\xi)+C_0(\tau,\xi)&=&(v_t+\ka m\tau)(\ka-R(\xi))+2\ka m\ln(1-D(\xi))+O(\xi^{-1})\\\nonumber
&=&-(v_t+\ka m\tau)\left(\sg_0(1-\rho^2)^{1/2}\xi-\ka+i\frac{\sg_0^2-2\rho\ka}{2\sg_0(1-\rho^2)^{1/2}}\right)
\\\nonumber
&&+2\ka m\ln(1-D(\xi))+O(\xi^{-1})\\
\label{exp1Dxi}
e^{2\ka m\ln(1-D(\xi)}&=&\left(\frac{2}{1-i\rho/(1-\rho^2)^{1/2}}\right)^{2\ka m}.
\eqa
It follows that, as $\xi\to\infty$ in the right half-plane, the integrand on the RHS of \eqref{callprice} has
the following asymptotics:
\bbe\label{asHeston}
C_\infty\frac{e^{iz_t\xi-c_\infty(0) \xi}}{\xi^2}(1+O(\xi^{-1})),
\ee
where
\beqa\label{CinfH}
C_\infty&=&\frac{Ke^{-r\tau}}{2\pi}\left(\frac{2}{i\rho/(1-\rho^2)^{1/2}-1}\right)^{2\ka m}\\\nonumber
&& \cdot\exp\left[(v_t+\ka m\tau)
\left(\ka-i\frac{\sg_0^2-2\rho\ka}{2\sg_0(1-\rho^2)^{1/2}}\right)\right],
\\\label{cinfH}
c_\infty(0)&=&(v_t+\ka m\tau)\sg_0(1-\rho^2)^{1/2}.
\eqa
Note that
\bbe\label{Cinfty2}
|C_\infty|=\frac{Ke^{-r\tau}}{2\pi}e^{(v_t+\ka m\tau)\ka}\left(4(1-\rho^2)\right)^{\ka m}.
\ee
Set  $\varphi_0=-\arctan(z_t/c_\infty(0))$, and
\begin{enumerate}[(i)]
\item
$\gam=-\pi/2-\varphi_0$, $\gap=\pi/2$, if $z_t>0$ (hence, $\varphi_0<0$),
\item
$\gam=-\pi/2, \gap=\pi/2-\varphi_0$, if $z_t<0$ (hence, $\varphi_0>0$).
\end{enumerate}
Thus, $\gam\in [-\pi/2,0), \gap\in (0,\pi/2]$. We define $\om$ and $d_0$ by \eqref{omd0},
then $\om+d_0=\gap, \om-d_0=\gam$. 
Next, we must ensure that the intersection of the imaginary axis and the image of $S_{(-d_0,d_0)}$ under $\chi_{\om_1, \om; b}$
is a subset of $i(\mum,\mup)$, which is equivalent to
$
\om_1+ba_+\le \mup, \om_1-b a_- \ge\mum,$ where
$a_-=-\sin\gam, a_+=\sin \gap$. Hence, we define $\om_1$ and $b_0$ by \eqref{om1b0}.
We choose $d<d_0$, $b<b_0$ close to $d_0, b_0$, respectively,
e.g., $d=0.95 d_0, b=0.95d_0$, and, for the given error tolerance $\eps$, set
$\ze =2\pi d/(\ln(H(f,d)/\eps)\sim 2\pi d/E$, where $E=\ln(1/\eps)$.

The approximate bound for the truncation error is
\bbe\label{truncerror222}
Err_{tr}(\La_1)\le 2|C_\infty|e^{-(z_t\sin\om+c_\infty(0)\cos\om)\La_1}/\La_1.
\ee
Given $\eps>0$, we find a moderately accurate approximation to the solution of the equation
\bbe\label{La1nu11}
(z_t\sin\om+c_\infty(0)\cos\om)\La_1+\ln\La_1-\ln(|C_\infty|/\eps)=0,
\ee
and then calculate
$\La=\ln(2\La_1/b),\ N=\mathrm{ceil}(\La/\ze).
$

\subsection{Complexity of the scheme} As $\eps\downarrow 0$, $\La\sim \ln E,$ where $E=E(\eps)=\ln(1/\eps)$, and
$\ze\sim E/(2\pi d)$, where $d<(\gap-\gam)/2$ is fixed. Hence, the complexity of the scheme is of the order of
$A(d)E\ln E$, where $A(d)=1/(2\pi d)<2/\pi^2$ if $d$ is chosen sufficiently close to $(\gap-\gam)/2>\pi/4$.


\subsection{Numerical results}
 In Section \ref{tableHeston}, Tables 3-8, we produce results for European put in the Heston model, and compare the performance
 of the sinh-acceleration with the fractional-parabolic change of variables. We adjust the recommended $\ze$ and $\La=N\ze$ dividing
$\ze$ by $k_\ze$ and multiplying $\La$ by $k_\La$. We show $\ze$, $N$ and the resulting errors for each choice of $\eps, k_\ze$ and $k_\La$.
The errors (rounded)
are calculated with respect to the benchmark prices (rounded). The latter are obtained using several sets of the parameters of the numerical scheme;
the results differed by less than  E-13. 
In Table 9, we compare the performance of the sinh-acceleration method with the Lewis-Lipton and Carr-Madan realizations of the flat iFT method.
 In all cases, the standard prescriptions ($\ze=0.125$, $N=4096$) imply negligible  truncation errors, hence,
the non-negligible errors shown are, essentially, the
discretization errors.  

\section{Options on bond in the CIR model}\label{CIRbond}
\subsection{Characteristic function} In the CIR model, the state space is $\bR_+$, the dynamics of the short rate is given by
\begin{equation}\label{SDECIR}
dr_t=\ka(\theta-r)dt+\sg\sqrt{r_t}dW_t,
\end{equation}
where $\ka,\theta, \sg>0$, and $dW_t$ is the increment of the standard Wiener process. For $t<T$ and $r>0$,
the characteristic function
\[
W(t,T; r,\xi)=\bE^{\bQ,r}_t\left[\exp\left(-\int_t^T r_sds\right)e^{i\xi r_T}\right],\ \xi\in \bR,
\]
is of the form
\begin{equation}\label{WCIR}
W(t,T;r,\xi)=\exp[B(\tau,\xi)r+C(\tau,\xi)],
\end{equation}
where $\tau=T-t$, and $B$, $C$ can be found solving the system of Riccati equations associated with the model.
Below, we reproduce
the well-known solution 
in the form convenient for application of the sinh-acceleration method:
\beqa\label{solBCIR2}
B(\tau,\xi)&=&\frac{B_+B_-+i\xi B_{+,n}(\tau)}{B_{++}(\tau)-i\xi},\\\label{solCCIR}
C(\tau,\xi)
&=&\ka\theta\left[B_-\tau+\frac{2}{\sg^{2}}\ln\frac{B_+-B_-}{1-e^{-\tau \sqrt{\ka^2+2\sg^2}}}-\frac{2}{\sg^{2}}\ln(B_{++}(\tau)-i\xi)\right],
\eqa
where $B_\pm=(\ka\pm\sqrt{\ka^2+2\sg^2})/\sg^2$,
\beqa\label{Bpn}
B_{+,n}(\tau)&=&\frac{B_+e^{-\tau \sqrt{\ka^2+2\sg^2}}-B_-}{1-e^{-\tau \sqrt{\ka^2+2\sg^2}}},\\\label{Bpp}
B_{++}(\tau)&=&\frac{B_+-B_-e^{-\tau \sqrt{\ka^2+2\sg^2}}}{1-e^{-\tau \sqrt{\ka^2+2\sg^2}}}.
\eqa

\begin{lem}\label{ChExp}
\begin{enumerate}[a)]
\item The characteristic function is analytic  in $\bC\setminus i(-\infty,-B_{++}(\tau)]$.
\item
As $\xi\to\infty$ remaining in $\bC\setminus i(-\infty,-B_{++}(\tau)]$,
\bbe\label{CIRasymp}
\left|e^{B(\tau,\xi)r +C(\tau,\xi)}\right|\sim C_\infty(\tau,r) |\xi|^{-2\ka\theta/\sg^2},
\ee
where
\bbe\label{CinfCIR}
C_\infty(\tau,r)=\left(\frac{B_+-B_-}{1-e^{-\tau\sqrt{\ka^2+2\sg^2}}}\right)^{2\ka\theta/\sg^2}
\exp\left[r\frac{B_+-B_-e^{\tau\sqrt{\ka^2+2\sg^2}}}{1-e^{\tau\sqrt{\ka^2+2\sg^2}}}+\ka\theta B_-\tau\right].
\ee
\end{enumerate}
\end{lem}
\begin{proof}
 a) It follows from \eqref{solBCIR2} and \eqref{solCCIR} that $B(\tau,\xi)$ is meromorphic in $\bC$ with the only simple pole at
 the root of the equation
\[
i\xi-B_--(i\xi-B_+)e^{\tau \sqrt{\ka^2+2\sg^2}}=0.
\]
The root is $-B_{++}(\tau)$, and $C(\tau,\xi)$ is analytic in $\bC\setminus i(-\infty, B_{++}(\tau)]$.

 b) As $\xi\to\infty$, $\La(\xi)\to 1$, hence, \eqref{CIRasymp}-\eqref{CinfCIR} follow from \eqref{solBCIR2} and \eqref{CinfCIR}.
 \end{proof}
\subsection{The bond price} We let $\xi=0$ in the formula for the characteristic function
\[
P(T, r)=\exp[B(T,0)r+C(T,0)].
\]
Since $\La(0)=B_+/B_-$, we have
\beqast
B(T,0)&=& \frac{B_+-B_-(B_+/B_-)e^{T \sqrt{\ka^2+2\sg^2}}}{1-(B_+/B_-)e^{T \sqrt{\ka^2+2\sg^2}}}
=
B_-\frac{e^{T \sqrt{\ka^2+2\sg^2}}-1}{e^{T \sqrt{\ka^2+2\sg^2}}-B_-/B_+}\\&=&B_-\frac{1-e^{-T\sqrt{\ka^2+2\sg^2}}}
{1-(B_-/B_+) e^{-T\sqrt{\ka^2+2\sg^2}}},\\
C(T,0)
&=&\ka\theta\left[B_-T+2\sg^{-2}\ln\frac{1-B_-/B_+}{1-(B_-/B_+)e^{-T \sqrt{\ka^2+2\sg^2}}}\right].
\eqast
\subsection{Call and put options}
Consider now the call option with the maturity date $\tau$ and strike $K<e^{C(T,0)}$, on the bond maturing at $T+\tau$.
Set $z_{T,K}=(C(T,0)-\ln K)/B(T,0)$. Since $B(T,0)<0$, the Fourier transform
of option's payoff
\beqast
G(\xi)&=&\int_\bR e^{-ir\xi}(e^{B(T,0)r+C(T,0)}-K)_+dr
=\frac{KB(T,0)e^{iz_{T,K}\xi}}{\xi(\xi+iB(T,0))},
\eqast
is well-defined in the half-plane $\{\Im\xi>-B(T,0)\}$, and admits meromorphic
continuation to the complex plane with two simple poles at 0 and $-iB(T,0)$. The characteristic function
$\exp[B(\tau,\xi)r+C(\tau,\xi)]$ of $r_\tau\vert r$ admits analytic continuation to the complex plane with the cut
$i(-\infty, -B_{++}(\tau)]$, and decays as $|\xi|^{-2\ka\theta/\sg^2}$ as $\xi\to\infty$ in the complex plane with the  cut.

Hence,  the price of the call option on the bond can be calculated as
\bbe\label{bondcallprice}
V_{\mathrm{call}}(\tau, r)=\frac{KB(T,0)}{2\pi}\int_{\Im\xi=\om_0}\frac{e^{iz_{T,K}\xi+B(\tau,\xi)r+C(\tau,\xi)}}{\xi(\xi+iB(T,0))}d\xi,
\ee
where $\om_0>-B(T,0)$ is arbitrary. The integrand  is meromorphic in $\bC\setminus i(-\infty, -B_{++}(\tau)]$, with two simple poles at 0 and $-iB(T,0)$,
hence, we have 3 strips of the analyticity $S_{(-B_{++},0)}$, $S_{(0,-B(T,0))}$ and $S_{(-B(T,0),+\infty)}$
of the integrand.
We can move the line of integration into any strip.
On crossing one or two poles, we apply the residue theorem.
When both poles are crossed, we obtain the price of the put. Thus,
\bbe\label{putcall}
V_{\mathrm{call}}(\tau, r)=V_{\mathrm{put}}(\tau, r)+e^{z_{T,K}B(T,0)}P(T+\tau,r)-KP(\tau,r),
\ee
which can be used to  double-check the accuracy
of calculations. Indeed, if $V_{\mathrm{call}}(\tau, r)$ and $V_{\mathrm{put}}(\tau, r)$ are calculated directly, with no pole crossed,
then a random agreement between the two completely differently sums, with the difference of the order of E-12,
 say, has a negligible probability unless the errors of both results are of the same order of magnitude.

 Let $S_{(\mum,\mup)}$ be the chosen strip. According to the general scheme, the choice of the parameters of the sinh-acceleration
 depends on the sign of $z_{T,K}$.
 If $z_{T,K}=0$, we can apply the sinh-acceleration with $\gam=-\pi/2, \gap=\pi/2$;
if $z_{T,K}>0$ with $\gam=0, \gap=\pi/2$; and if $z_{T,K}<0$, with $\gam=-\pi/2, \gap=0$. If $z_{T,K}\ge 0$, then,
for any $r\ge 0$, $\exp[B(T,0)r+C(T,0)]\le K$, hence, the price of the call option is 0, and it is unnecessary to apply
a numerical method to recover this zero. However, to test the accuracy of the method, we applied our general recommendations to
this case as well, and, in numerical examples, verified that the call option price calculated numerically is of the order of the
error tolerance used to choose the parameters of the scheme.

We define $\om$ and $d_0$ by \eqref{omd0},
and $\om_1$ and $b_0$ by \eqref{om1b0}. We choose $d<d_0$, $b<b_0$ close to $d_0, b_0$, respectively,
e.g., $d=0.95 d_0, b=0.95d_0$.
Then, after the change of variables \eqref{sinhbasic}, the integrand in the pricing formula,
denote it $f(y)$, admits analytic continuation to the strip $S_{(-d,d)}$ and decays sufficiently fast as $y\to\infty$ remaining
in the strip. The Hardy norm \eqref{Hnorm} is finite and can be approximated well by \eqref{hfdbound}.
To satisfy a small error tolerance $\eps>0$,
we choose $\ze =2\pi d/(\ln(H(f,d)/\eps)\sim 2\pi d/E$, where $E=\ln(1/\eps)$.

 The truncation error of the simplified trapezoid rule
can be approximated by the truncation error of the integral \eqref{truncerror}.
The rate of decay of the integrand is the smallest if $z_{T,K}=0$. In this case, the integrand decays as
\[
\frac{K(-B(T,0))C_\infty(\tau,r)}{2\pi}|\xi|^{-2-2\ka\theta/\sg^2}\]
(see  \eqref{CIRasymp}-\eqref{CinfCIR}), therefore, given the error tolerance $\eps$, we can find $\La_1=b e^{\La}/2$ from
\[
\frac{K(-B(T,0))C_\infty(\tau,r)}{\pi}\int_{\La_1} y^{-2-2\ka\theta/\sg^2}dy=\eps,\]
which gives $\La_1=\eps_1^{-1/(1+2\ka\theta/\sg^2)}$, where $\eps_1=\eps(1+2\ka\theta/\sg^2)/(K(-B(T,0))C_\infty(\tau,r))$.
Thus,
\[
\La=\ln (2/b)+(1+2\ka\theta/\sg^2)^{-1}\ln(1/\eps_1),\ N=\mathrm{ceil} (\La/\ze).\]

 If $z_{T,K}\neq 0$, then $|\om|=\pi/4$, and $\Re (iz_{T,K}\xi)\sim -c_\infty(T,K)|\xi|$ along the contour
 $\xi=\chi_{\om_1,\om,b}(\bR)$, where $c_\infty(T,K)=|z_{T,K}\sin\om |$. Hence, we need to find $\La_1=b e^{\La}/2$ satisfying
\[
\frac{K(-B(T,0))C_\infty(\tau,r)}{\pi}\int_{\La_1} e^{-c_\infty(T,K)y}y^{-2-2\ka\theta/\sg^2}dy\le \eps.\]
We solve a stronger equation
\[
e^{-c_\infty(T,K)\La_1}\La_1^{-1-2\ka\theta/\sg^2}=\eps_1,
\]
equivalently,
\[
F(\La_1):=c_\infty(T,K)\La_1+(1+2\ka\theta/\sg^2)\ln \La_1-\ln(1/\eps_1)=0,
\]
as follows: $\La_1=(1/c_\infty(T,K))\ln(1/\eps_1)$,
\[
\La_1:=\La_1-\frac{1+2\ka\theta/\sg^2}{c_\infty(T,K)}\ln \La_1.
\]
Finally, we calculate $\La=\ln(2\La_1/b)$, $N=\mathrm{ceil} (\La/\ze).$

\begin{rem}{\rm
\begin{enumerate}
\item
If $z_{T,K}\neq 0$ (and not too small in absolute value), then the rate of decay is, essentially, as in the case of the call option
in regular SINH-models of order $\nu\in (0,1)$ and non-zero $x'$. However, if $z_{T,K}$ is zero or close to zero, then
the number of terms needed to satisfy a given error tolerance can be larger - and very large if the sinh-acceleration is not used.
Even the fractional-parabolic deformation requires 10 times more terms (for some parameters, even more) to achieve the same accuracy.
\item
Formally, one should use the widest strip $S_{(-B(T,0),+\infty)}$ and choose an arbitrary large $\om_0>-B(T,0)$. However, if $\om_0$
is large, then the simplified general recommendations for the choice of $\ze$ and, especially, $\La$ can become too inaccurate. Indeed,
if  $z_{T,K}<0$, then the wings of the curve $\cL_{\om_1,\om; b}:=\chi_{\om_1, \om; b}(\bR)$ point downward,
and we can truncate the sum in the infinite trapezoid rule where the integrand becomes sufficiently small.
However, if $\om_0>0$ is not small, a significant number of points $\xi_j=\chi_{\om_1,\om;b}(y_j)=i\om_1+b\sinh(i\om +y_j)$ used in the simplified trapezoid rule
are in the upper half-plane but the simplified recommendation implicitly presumes that all the points are in the low half-plane.

Thus, if $\om_0>0$ and $z_{T,K}<0$, we need to modify the prescription above by adding to $\La$ the half-length $\La_0$ of the intersection of the curve
$\cL_{\om_1,\om; b}$ with the upper half-plane.  To find $\La_0$, we solve the equality
$
\om_1+b\Im\sinh(i\om+y)=0$, equivalently, $e^y-e^{-y}+2\om_1/(b\sin\om)=0$. Thus,
\[
\La_0=-\om_1/(b\sin\om)+\sqrt{(\om_1/b\sin\om)^2+1)},
\]
and $\La=\ln(2\La_1/b)+\La_0$, $N=\mathrm{ceil} (\La/\ze).$ This increases the number of terms. Hence, it is advisable
 to choose $\om_0\in (-B_{++},0)$ unless $B_{++}$ is very small.
\end{enumerate}
}
\end{rem}

\subsection{Numerical example}  Table 10 in Section \ref{CIRTable} demonstrates that
the sinh-acceleration is significantly faster than the fractional-parabolic method (the number of terms is 10-30 times fewer
and the CPU time in Matlab realization is about 5 times smaller); the flat iFT can satisfy the error tolerance $10$ mln times larger
only when the number of terms is of the order of $10^5$, and the CPU time is 100 times larger.

\section{Subordination}\label{subord}
We consider the following example. Let $y_t$ be the square root process
with the dynamics
\bbe\label{CIRdyn}
dy_t=\ka(\theta-y_t)dt+\la \sqrt{y_t}dW_t,
\ee
where $\ka>0, \la>0, \theta>0$ and $dW_t$ is the increment of the standard Wiener process. A popular subordinator
$
Y_t=\int_0^t y_s ds
$
conditioned on $y_0$ has the characteristic function
\beqa\nonumber
\Phi_{\mathrm{CIR}}(t, y_0; \eta)&:=&\bE\left[ e^{i\xi Y_t}\ |\ Y_0=y_0\right]\\\label{chfYt}
&=&\frac{\exp(\ka^2\theta t/\la^2)\exp(2y_0i\eta/(\ka+u(\eta)\coth(u(\eta) t/2)))}
{\left[\cosh(u(\eta) t/2)+\ka \sinh(u(\eta) t/2)/u(\eta)\right]^{2\ka\theta/\la^2}},
\eqa
where $u(\eta)=\sqrt{\ka^2-2\la^2i\eta}$. The pdf of $Y_t | Y_0=y_0$ is
\bbe\label{pdfCIR}
p(y_0;\tau)=\frac{1}{2\pi}\int_\bR e^{-i\tau\eta}\Phi_{\mathrm{CIR}}(t, y_0; \eta)d\eta.
\ee
Since $\tau>0$, and $\Phi_{\mathrm{CIR}}(t, y_0; \eta)$ is uniformly bounded on the domain of analyticity, we must 
use a cone in the lower half plane $\Im\eta<0$; $u$ being analytic in the complex plane with the cut
$i(-\infty,-\ka^2/(2\la^2)]$, the simplest choice is $\gam=-\pi/2, \gap=0$, $\om=-\pi/4$, $d_0=\pi/4$.
\begin{lem}\label{lem:CIRsub}
\begin{enumerate}[a)]
\item
For $\eta\in \bC\setminus i(-\infty,-\ka^2/(2\la^2)]$,  $\Re u(\eta)>0$.
\item
Fucntion
$\bR\ni\eta\mapsto \Phi_{\mathrm{CIR}}(t, y_0; \eta)\in\bC$ admits analytic continuation to $\bC\setminus i(-\infty,-\ka^2/(2\la^2)]$.
\end{enumerate}
\end{lem}
\begin{proof} a) is evident. b) We set $\ga=2\la\theta/\la^2$, $w=w(\eta)=e^{-u(\eta) t/2}$, and rewrite the denominator on
the RHS of \eqref{chfYt} as
\beqa\label{acontsqbr}
&&
\left[\cosh(u(\eta) t/2)+\ka \sinh(u(\eta) t/2)/u(\eta)\right]^{2\ka\theta/\la^2}
\\\nonumber
&=&\frac{e^{\ga u(\eta) t/2}}{(2u(\eta))^\ga}(1+w^2)^\ga\left(u(\eta)+\ka\frac{1-w^2}{1+w^2}\right)^\ga.
\eqa
Since $\Re u(\eta)>0$,  $w_1:=w(\eta)^2=e^{-u(\eta) t/2}$ belongs to the unit open disc $\cD$, hence,
the fraction on the RHS of \eqref{acontsqbr} and the factor $(1+w^2)^\ga$ are well-defined analytic functions
on  $\bC\setminus i(-\infty,-\ka^2/(2\la^2)]$. To prove that the same statement holds for the last factor, it suffices to
show that $(1-w^2)/(1+w^2)$ is in the right half-plane if $w^2\in \cD$. Let $w^2=a+ib$, where $a^2+b^2<1$. Then
\[
\Re \frac{1-w^2}{1+w^2}=\Re\frac{(1-a)-ib}{1+a+ib}=\Re\frac{((1-a)-ib)(1+a-ib)}{(1+a)^2+b^2}=\frac{1-a^2-b^2}{(1+a)^2+b^2}>0.
\]
\end{proof}
The pricing formula for European options in the CIR-subordinated L\'evy models changes as follows.
Instead of the expectation $\bE[e^{i\xi X_\tau}\ |\ X_0=x] = e^{ix\xi-\tau\psi(\xi)}$, we have the expectation
$\bE[e^{i\xi X_{Y_\tau}}\ |\ X_0=x, Y_0=y_0] = e^{ix\xi}\Phi_{\mathrm{CIR}}(\tau, y_0; i\psi(\xi))$. Hence,
\bbe\label{EuroCIR}
V_{\mathrm{call}}(S,K; y_0,\tau)=-\frac{Ke^{-r\tau}}{2\pi}
\int_{\Im\xi=\om_0}\frac{e^{ix\xi}\Phi_{\mathrm{CIR}}(\tau, y_0; i\psi(\xi))}{\xi(\xi+i)}d\xi,
\ee
where $x=\ln(S/K)$ and $\om_0<-1$ is such that, for all $\xi$ on the line $\Im\xi=\om_0$,  $i\psi(\xi)$ is in the complex plane with the cut
$i(-\infty,-\ka^2/(2\la^2)]$, equivalently, $\psi(\xi)$ is in the complex plane with the cut $(-\infty, -\ka^2/(2\la^2)]$.
This implies that the CIR subordinator must satisfy the condition $\ka^2/(2\la^2)>-\psi(-i)$.

Applying the sinh-acceleration, we need to choose the parameters of the scheme so that
\begin{enumerate}[(1)]
\item
the image of the strip $S_{(-d_0,d_0)}$ under the composition $y\mapsto \psi(\chi_{\om_1; \om, b}(y))$
 belongs to the complex plane with the cut $(-\infty, -\ka^2/(2\la^2)]$ or to an appropriate Riemann surface;
\item
if $x>0$, then $\chi_{\om_1; \om, b}(S_{(-d_0,d_0)})$ must be a subset of a half-plane of the form $\Im\xi>a$, where $a\in\bR$ is a constant;
\item
if $x<0$, then $\chi_{\om_1; \om, b}(S_{(-d_0,d_0)})$ must be a subset of a half-plane of the form $\Im\xi<a$, where $a\in\bR$ is a constant.
\end{enumerate}

Furthermore, in cases when $\Phi_{\mathrm{CIR}}(\tau, y_0; i\psi(\xi))$ admits analytic continuation to an appropriate Riemann surface,
it is advisable to choose the parameters of the sinh-acceleration so that the deformed contour $\cL_{\om_1;\om, b}=\chi_{\om_1;\om,b}(\bR)$ remains
in the complex plane and additional operations in the program caused by appropriate phase shifts are not needed. If the deformed contour crosses a
cut, the number
of terms in the simplified trapezoid rule decreases somewhat but the number of elementary operations needed to evaluate individual terms
increases. The total gain is small if any.

Let $X$ be an elliptic sinh-regular
process of type $((\lm, \lp); \cC, \cC_+)$ and order $\nu\in (0,2]$, where
and $\lm<0<\lp$; furthermore, as $\xi\to \infty$ remaining in the conus $\cC$,
\[
\psi^0(\xi)\sim c_\infty e^{i\varphi\nu}|\xi|^\nu,\ \xi\to+\infty,
\]
where $\varphi=\mathrm{arg}\, \xi$, and $c_\infty>0$.

First, we find a strip $S_{(\mum,\mup)}$ of analyticity of $\Phi_{\mathrm{CIR}}(t, y_0; i\psi(\xi))$. Here
 $\mum<0<\mup$ are such that $\ka^2+2\la^2\psi(i\mu)>0$ for all $\mu \in (\mum,\mup)$.
Since $\mu\mapsto \psi(i\mu)$ is convex on $(\lm,\lp)$ and $\psi(0)=0$, we conclude that  if
$\ka^2+2\la^2\psi(i(\lp-0))\ge 0$, then $\mup=\lp$, otherwise $\mup$ is the only positive solution of the equation
$\ka^2+2\la^2\psi(i\mu)=0$. Similarly, if
$\ka^2+2\la^2\psi(i(\lm+0))\ge 0$, then $\mum=\lm$, otherwise $\mum$ is the only positive solution of the equation
$\ka^2+2\la^2\psi(i\mu)=0$.

Next, we need to find a conus of analyticity, and calculate the asymptotics of $\Phi_{\mathrm{CIR}}(t, y_0; i\psi(\xi))$
as $\xi\to \infty$ remaining in the conus. We consider two cases.
\sbr
\noindent
{\em Case Ia.} Let $\nu\in (1,2]$ or $\nu\in (0,1)$ and $\mu=0$. Then $\psi(\xi)=i\mu\xi+\psi^0(\xi)$ enjoys the same properties as $\psi^0$, and
\[
\psi(\xi)\sim c_\infty e^{i\varphi\nu}|\xi|^\nu, \ \xi\to+\infty.
\]
\sbr
\noindent
{\em Case Ib.} Let  $\nu=1$. Then the asymptotic coefficient and argument depend on $\mu$:
\[
\psi(\xi)\sim c_\infty(\mu) e^{i(\varphi+\ga(\mu))}|\xi|,\ \xi\to+\infty.
\]
Thus, there exist $-\pi\le \gam<0<\gap\le  \pi$
  such that for any $\varphi\in (\gam,\gap)$,
\[
\psi(\rho e^{i\varphi})\sim c_\infty(\varphi)\rho^\nu,\ \rho\to+\infty,\]
where $c_\infty(\varphi)\not\in (-\infty,0]$. Hence, for any $\varphi\in (\gam,\gap)$,
\[
(\ka^2+2\la^2\rho e^{i\varphi})^{1/2}\sim c_\infty(\varphi)^{1/2}\rho^{\nu/2}, \ \rho\to+\infty,
\]
where $\Re c_\infty(\varphi)^{1/2}>0$.

The argument above ``almost" proves that $(\ka^2+2\la^2\psi(\xi))^{1/2}$ admits analytic continuation to
$i(\mum,\mup)+\cC_{\gam,\gap}$. We say almost because the proof above demonstrates that, for
$\xi\in i(\mum,\mup)+ \cC_{\gam,\gap}$, $\ka^2+2\la^2\psi(\xi)\not\in (-\infty,0]$ if $\xi$ is in a certain neighborhood of 0 and
a certain neighborhood of infinity. For NTS and KoBoL of order $\nu\in [1,2]$, one can demonstrate that the image
of  $i(\mum,\mup)+ \cC_{\gam,\gap}$ under the map $\xi\mapsto \ka^2+2\la^2\psi(\xi)$ does not intersect $(-\infty, 0]$;
in general case, one should study the image on the case-by-case basis, and, if necessary, use $\mu_\pm$ closer to 0.
Note that it suffices to ensure that the image does not intersect the essentially singular point 0, which is a much weaker condition. If the image intersects $(-\infty,0)$ but does not contain 0, 
the image is a subset of an appropriate Riemann surface, and a larger $\ze$ can be chosen. However, it is advantageous to choose the parameters
of the sinh-acceleration so that the image of the deformed contour under the map $\xi\mapsto \ka^2+2\la^2\psi(\xi)$ is a subset
of the complex plane, and there is no need to introduce phase shifts in the pricing formula, when the cut is crossed.

Once $i(\mum,\mup)+ \cC_{\gam,\gap}$ is found, we define the parameters of the sinh-acceleration and choose $\ze$ for
the given error tolerance using the general prescriptions. It remains to find the truncation parameter and $N$.
If follows from \eqref{LaN} that as $\xi=\rho e^{i\varphi}\to \infty$ remaining in $\cC_{\gam, \gap}$,
\[
|\Phi_{\mathrm{CIR}}(t, y_0; i\psi(\xi))|\le (1+o(1))\exp(\ka^2\theta t/\la^2-B(\varphi)\rho^{\nu/2}).
\]
where
\[
B(\varphi)=\sqrt{2}\la \Re(c_\infty(\varphi\nu)^{1/2})(t/2)2\ka\theta/\la^2=\sqrt{2} \Re(c_\infty(\varphi\nu)^{1/2})t\ka\theta/\la.
\]
Hence, we can find the truncation parameter $\La=\ln(2\La_1/b)$ solving approximately the equality
\[
e^{-A(\om,x)\La_1-B(\om)\La_1^{\nu/2}}\rho^{-1}=\eps_1,\]
where $A(\om,x)=|x\sin\om|$, $\eps_1=\eps\pi\exp(r\tau-\ka^2\theta t/\la^2)/K$, as follows. 

If $\nu=1$, then
\[
\La_1:=\ln(1/\eps_1)/(A(\om)+B(\om)),\ \La_1:=\max\{2, \La_1-\ln\La_1/(A(\om)+B(\om))\};\]
if $F:=(A(\om,x)+B(\om))\La_1+\ln\La_1-\ln(1/\eps_1)<0$, then
\[
\La_1:=\La_1-F/(A(\om)+B(\om)+1/\La_1).
\]
If $\nu\in (1,2]$ or $\nu\in (0,1)$ and $\mu=0$, then we set
\[
\La_1=(\ln (1/\eps_1)/B(\om)^{2/\nu},\ \La_1=\left(\frac{\max\{2, \ln(1/\eps_1)-A(\om)\La_1-\ln\La_1\}}{B(\om)}\right)^{2/\nu};\]
if $F:=A(\om,x)\La+B(\om)\La_1^{\nu/2}+\ln\La_1-\ln(1/\eps_1)<0$, then $\La_1:=\La_1-F/DF$, where
\[
DF:=A(\om,x)+(\nu/2)B(\om)\La^{\nu/2-1}+1/\La_1.
\]
Note that if $B(\om)$ is very small and $|x|$ is not very small, then it is safer to use the same rule as in the case $\nu=1$;
this  leads to a moderate overkill.

\sbr
\noindent
{\em Case II.} $\nu\in (0,1)$ and $\mu\neq 0$. If $x>0$, we take arbitrary $0<\gam<\gap<\pi/2$; as $\xi\to \infty$ remaining in $\cC_{\gam,\gap}$,
$\ka^2-2\la^2 i\psi(\xi)\sim \ka^2-2\la^2\mu\xi$, hence, if $\mu>0$, then
\[
u(i\psi(\rho e^{i\varphi}))\sim e^{i((\varphi-\pi/2)}\rho^{1/2},\ \rho\to +\infty,
\]
and if $\mu<0$, then
\[
u(i\psi(\rho e^{i\varphi}))\sim e^{i(\varphi/2)}\rho^{1/2},\ \rho\to +\infty.
\]
If $\varphi\in (0,\pi)$, $\cos (\varphi-\pi/2)>0$ and $\cos (\varphi/2)>0$. Hence, if $\nu\in (0,1)$ and $\mu\neq 0$,
the rate of the decay of the integrand is as in the case $\nu=1$ but the asymptotic coefficient is different.
We leave the details to the reader.

The results of a numerical experiment can be found in Table 11.

\section{Quantiles and  Monte-Carlo simulations}\label{s:MC}
\subsection{One-factor KoBoL}\label{ss:MC1KBL}
We consider evaluation of quantiles, that is, solution
of the equation $F(x)=A$, where $A\in (0,1)$ and $F$ is the cumulative
distribution function; once an efficient procedure for the quantile evaluation is available,
the procedure can be used in Monte-Carlo simulations.

Monte-Carlo simulation remains to this day the most universal method of pricing path-dependent options on financial assets. In the case of L\'evy-driven models, a basic building block of any Monte-Carlo method is the simulation of an increment of the underlying L\'evy process. In some situations --- for instance, the Variance Gamma model \cite{MS90,MM91,MCC98} --- the process can be expressed in terms of simpler ones using time subordination, and hence its increments can be simulated (almost) exactly. However, in other cases no exact simulation algorithm is known.  Madan and Yor \cite{madan-yor} proposed an algorithm for simulating KoBoL increments based on representing the process as Brownian motion subordinated by a stable L\'evy process, truncating the small jumps of the subordinator and replacing them with their average. However, as extensive numerical examples in \cite{MCMityaLevy}
demonstrate, an efficient implementation of the standard approach described below is 10-100 faster; the variation which we
introduce below, is much faster and more accurate than the realization in \cite{MCMityaLevy}.

If $Z$ is any random variable with continuous distribution, one can simulate $Z$ sampling a uniformly distributed random variable $U$ on $(0,1)$ and calculate $F^{-1}(U)$, where $F$ denotes the cumulative distribution function of $Z$. When an explicit formula for $F^{-1}$ is not available, it becomes important to be able to calculate its values very quickly and sufficiently accurately. A straightforward approach that was used with a limited success for simulation
of stable L\'evy processes (the tails of the distributions decay too slowly, hence, the Monte-Carlo simulations are moderately efficient only
if the index of the process is close to 2, and the distribution does not differ much from the normal distribution, with the exception
of far parts of the tails, which can be safely disregarded in this case) is as follows. One tabulates
the values of $F$ on a sufficiently long and fine grid of points $x_0,x_1,\dotsc,x_M$ on the real line and approximates $F^{-1}$ using linear interpolation.
This method is very attractive from the practical viewpoint, because the values $F(x_i)$ only have to be calculated once, and afterward the computational cost of each simulation of $Z$ is extremely low: one has to draw a sample $A$ of $U$, find $j$ satisfying $F(x_j)\leq A<F(x_{j+1})$ (which requires about $\log_2(M)$ comparisons)
and perform 4--5 arithmetic operations required for linear interpolation:
\begin{equation}\label{linintnon-log}
x=x_j+(x_{j+1}-x_j)(A-F_j)/(F_{j+1}-F_j).
\end{equation}
If $A<F_1$, one assigns $x=x_1$, and if $A>F_M$, then
one assigns $x=x_M$.
In many important examples, an explicit formula is known for the characteristic function of the random variable $Z$.
In such a case, the calculation of the values $F(x_i)$ reduces to computing certain inverse Fourier transforms
(see Glasserman and Liu \cite{glass-liu-07,glass-liu-10}). In the case of L\'evy processes with exponentially decaying tails,
the problem of a slow decay is less serious than in the case of stable L\'evy processes unless the exponential rate is too small but the peak of the probability distribution remains very high if the order of KoBoL is close to $0$.

In the application to the Monte-Carlo simulations, this method  has 3 sources of errors:
\begin{enumerate}[(1)]
\item
truncation error;
\item
errors of linear interpolation;
\item
errors of evaluation of $F_k$.
\end{enumerate}
A popular approach (see, e.g., \cite{glass-liu-07,GlassermanLiu07,glass-liu-10,ChenFengLin,CGMYsim,feng-lin11}) is to use either the fast Fourier transform (FFT) or
fast Hilbert transform (fast HT), which allows one to calculate the values $F_k$ at
all points of interest faster than point-by-point, especially if the number of points is large.
This approach faces the following fundamental difficulties:
\begin{enumerate}[(a)]
\item
if the time step is small, which is necessary for accurate simulation of barrier options
with continuous monitoring, then, in a neighborhood of $x'=0$, the derivative $F'(x)=p(x)$ is very large,
hence, in order that the linear interpolation \eqref{linintnon-log} be sufficiently accurate,
the size of the mesh $\Delta x=x_{j+1}-x_j$ must be very small; in fact, even in a relatively
nice case  $\nu\in (0.5, 1)$, $\Delta x=10^{-5}$ can lead to interpolation errors
greater than $10^{-8}$ (see an example below); for $\nu$ closer to 0, much smaller $\Delta x$
can be insufficient;
\item
if one of the steepness parameters $\la_\pm$ is small in
absolute value, e.g., $\lp>0$ is small, then $x_1$ must be negative and very large
in absolute value to ensure that the truncation error in the neighborhood of $-\infty$
is sufficiently small. In view of difficulty (a), the total number of points can be
measured in millions and dozens of million;
\item
as examples in Section 2 demonstrate, accurate evaluation of $F(x_k)$ for $x_k$ large
in absolute value can be either too time consuming or virtually impossible if the flat FT is used;
the same difficulties arise is HT is used.
\end{enumerate}
The sinh-acceleration allows us to calculate $F_k=F(x_k)$ very accurately and fast;
for $x_k$ large in absolute value, especially fast. The fractional-parabolic method used
in \cite{MCMityaLevy} to calculate $F_k$ is faster and more accurate than FFT and HT-based methods
but, after the fractional-parabolic change of variables the number of terms $N$ in the simplified trapezoid rule depends on $x_k$ much stronger than
after the sinh-acceleration. When the latter is applied,
given the error tolerance,
 the parameters of the sinh-acceleration procedure
bar the number of terms $N$ can be chosen the same for all $x'\le 0$; similarly, for $x'\ge 0$. Furthermore, $N$
decreases as $x'$ increases in the absolute value. Hence, the parameters of the sinh-acceleration and $\ze, N$ can be chosen
for $x_1$ (resp., for $x_M$), and used to evaluate $p(x), p'(x), F(x)$ for all $x\le x_1$ (resp., for all $x\ge x_M$).
If $F^{-1}(a)<x_1$, we  calculate $x=F^{-1}(a)$ using the Newton method, and 2-3 iterations suffice
to satisfy the error tolerance $10^{-12}$ and less if the initial approximation $x_1<0$ is not small in absolute value.
Similarly, if $F^{-1}(a)>x_M$, we  calculate $x=F^{-1}(a)$ using the Newton method, and 2-3 iterations suffice
to satisfy the error tolerance $10^{-12}$ and less if the initial approximation $x_M>0$ is not small.
There is no need to truncate the state space.

To apply the Newton method
\begin{equation}\label{linN}
x_{n+1}=x_n-(F(x_n)-A)/F'(x_n),
\end{equation}
one has to evaluate the pdf $p_n=F'(x_n)$ as well but the sinh-acceleration
method allows one to calculate both $F(x_n)$ and $F'(x_n)$ using the same
parameters of the numerical scheme; moreover, only one step of calculations
is different: in the case of $F(x_n)$, we have an additional factor $1/(-i\xi)$.
We can simultaneously calculate $F^{\prime\prime}(x_n)$ inserting the factor $-i\xi$ instead of $1/(-i\xi)$;
this allows us to use the second order approximation
\[
F(x)=F(x_n)+(x-x_n)F'(x_n)+\frac{(x-x_n)^2}{2}F^{\prime\prime}(x_n)\]
to solve $F(x)-A=0$ on a sufficiently small interval $[x_{n-1}, x_n]$:
\begin{equation}\label{QT}
x=x_n-\frac{2(F(x_n)-A)}{F'(x_n)+\sqrt{F'(x_n)^2-2(F(x_n)-A)F^{\prime\prime}(x_n)}}.
\end{equation}
In general, we can calculate $F(x_n), F'(x_n), F^{\prime\prime}(x_n)$
for $x_n$ over a wide interval using the same parameters of the scheme. This implies that
the bulk of the CPU time is spent on calculations of the parameters of the scheme,
and the arrays $\xi_k=i\om_1+b\sinh(i\om+y_k)$ and $Exp_k:=\exp[-t\psi(i\om_1+b\sinh(i\om+y_k))]\cosh(i\om+y_k)$,
$k=1,2,\ldots, N$. The last step is the fast and straightforward evaluation
of the quantities
\begin{eqnarray}\label{F}
F(x_n)&=&\frac{ib\ze}{\pi}\sum_{j=0}^N (1-\de_{j0}/2)\Re \frac{\exp[-ix'_n\xi_j] Exp_j}{\xi_j}\\\label{Fp}
F'(x_n)&=&\frac{b\ze}{\pi}\sum_{j=0}^N (1-\de_{j0}/2)\Re \exp[-ix'_n\xi_j] Exp_j\\\label{Fpp}
F^{\prime\prime}(x_n)&=&\frac{ib\ze}{\pi}\sum_{j=0}^N (1-\de_{j0}/2)\Re \exp[-ix'_n\xi_j]\xi_j Exp_j.
\end{eqnarray}
For the application of the Newton method for $a<F_1$ and $a>F_M$, at each step of the iteration procedure,
we can use the arrays $\xi$ and $Exp$ calculated for $N=N(x_1)$ and $N=N(x_M)$, respectively.

The quadratic approximation \eqref{QT} allows us to use much sparser grids than
the linear approximations \eqref{linintnon-log} and \eqref{linN}, for $x_n$ not small in absolute value especially.

The next trick allows us to decrease the number of points smaller still. Instead of the equation
$F(x)=A$, we solve the equation $f(x)=a$, where $f(x)=\ln F(x)$ and $a=\ln A$. Since $f$ is more regular than $F$,
the same approximations work better, and $f'_k=f'(x_k)=F'_k/F_k$, $f^{\prime\prime}_k=(F^{\prime\prime}_kF_k-(F'_k)^2)/F_k$ are easy to calculate.

\begin{example}{\rm Consider KoBoL of order $\nu=0.7$ with 				
$c_+=c_+=0.6,$				
$\lp=5$,	$\lm=-10$, $\mu=0$; $t=0.001$. The second	instantaneous moment
	$m_2=	0.093440429$  (rounded) is not small, and time step $t=0.001$ is not exceedingly small.
The order $\nu=0.7$ is not small as well; in the empirical literature, one find numerous examples of $\nu$ close to 0.
The steepness parameter $\lp$ is not too small as well (one can find examples of $\lp<1$).
Nevertheless, as several examples for quantiles demonstrate,
\begin{enumerate}[1.]
\item
	accurate Monte-Carlo simulations using FFT or HT will require grids
with the size of the mesh $10^{-5}$ or less;
\item
 if the truncation
is made at the level $F_1=10^{-8}$, then $x_1=-1.6707581397416$ (the result is found using the Newton method
with the initial approximation $-1$; three iterations were needed to satisfy the error tolerance $10^{-12}$).
 Hence, FFT or HT method would require the uniformly spaced grid of the length of the order of 160k,
 and the errors of truncation and evaluation of $F_k$ would be non-negligible;
 \item
 with the exception of a very small neighborhood of 0, the quadratic approximation
 applied to $f$ requires much sparser grid than other approximations.
 \end{enumerate}
 In Table 12 (see Appendix \ref{examplesfractiles}), we list the errors of several approximations, for several values of $A$ and several
 widths $h$ of the interval $(x_{j-1}, x_j)$ containing $f^{-1}(a)=F^{-1}(A)$.
 }
\end{example}
Labels for approximations used:
\begin{itemize}
\item
L: linear interpolation \eqref{linintnon-log};
\item
N: Newton approximation \eqref{linN};
\item
LL: linear interpolation applied to $f=\ln F$;
\item
LN: Newton approximation applied to $f=\ln F$;
\item
QT: quadratic approximation \eqref{QT} applied to $f=\ln F$.
\end{itemize}
From Table 12, it is clearly seen that QT allows one to use much sparser grid $x_1<x_2<\cdots<x_M$; the grid must be non-uniformly
spaced. At the points of the grid, $f=\ln F$ and its first and second derivatives must be precalculated, which can be done very fast
using the sinh-acceleration. For evaluation of $x=f^{-1}(a)$ for $a<f_1$ and $a>f_M$, we use the Newton method and two
precalculated arrays of small sizes, which represent functions in the dual space. No truncation is needed.

\vskip0.2cm
{\em Outline of the algorithm.}

\begin{enumerate}[I.]
\item
 In a neighborhood of  $x'=0$, e.g., in the interval $[F^{-1}(0.3), F^{-1}(0.7)]$,
 the steps $h_j=h_{j+1}-h_j$ should be of the order $10^{-5}$ if $t$ is small. For larger $t$,
 larger steps are admissible. E.g., for $t=1$, $h_j$ of the order of 0.001 can be admissible.
 \item
 As $|x_j|$ increases, $h_j$ can be made larger. As the rule of thumb, for the left tail, we
 would recommend $h_j=-0.01 f_{j+1}/f'_{j+1}$ for points below $x_{low}:=F^{-1}(0.3)$; the points $x_j$ below $x_{low}$
 are calculated in the same cycle as the values  $f_j, f'_j, f^{\prime\prime}_j$. For the right tail, the recommendations are by symmetry.
 \item
 The grid is truncated at $F^{-1}(0.001)$ or $F^{-1}(0.0001)$, and at $F^{-1}(0.999)$ or $F^{-1}(0.9999)$.
 For all points of the grid, the values  $f_j, f'_j, f^{\prime\prime}_j$ should be calculated and stored.  \item
  The parameters of the sinh-acceleration
 should be calculated for $x_0=F^{-1}(0.001)$ (or $F^{-1}(0.0001)$), and arrays $\xi=\xi^-$ and $Exp^-$ calculated and stored.
 (The sign minus indicates that the arrays will be used for calculations in the left tail). 
 \item
  The parameters of the sinh-acceleration
 should be calculated for $x_M=F^{-1}(0.999)$ (or $F^{-1}(0.9999)$), and arrays $\xi=\xi^+$ and $Exp^+$ calculated and stored.
 (The sign plus indicates that the arrays will be used for calculations in the right tail).
 \item
 If a realization $A\sim U[0,1]$ belongs to $(0,F^{-1}(0.001)]$ (resp., to $[0,F^{-1}(0.999)]$,
 then an interval $[x_{n-1}, x_n]$ s.t. $f_{n-1}<\ln A\le f_n$ (resp., an interval $[x_n, x_{n+1}]$ s.t. $f_n\le A<f_{n+1}$)
 should be found, and the quadratic approximation \eqref{QT} applied.
 \item
 If $A<0.001$, the Newton method is applied with $x_0$ as the initial approximation;
 the stored values are used to calculate $f(x_n)/f'(x_n)$ at each step of the Newton method.
 \item
 If $A>0.999$, the Newton method is applied with $x_0$ as the initial approximation;
 the stored values are used to calculate $f(x_n)/f'(x_n)$ at each step of the Newton method.

 \end{enumerate}
 We call the
 arrays $\xi^\pm$ and $Exp^\pm$ {\em the conformal principal components}. The conformal principal components
 (evaluated at points of a grid different from the grids used for calculations in the tails) can be used to
 calculate quantiles $F^{-1}(A)$ for $A\in [0.001, 0.999]$. In this case, the bisection method instead of the Newton method
 should be used.

 \subsection{Monte-Carlo simulations in  regime-switching L\'evy models}\label{ss:MCRSwLevy}
 For each state of the modulating Markov chain, precalculate the pdf and cpf of the corresponding L\'evy process in an appropriate neighborhood of 0,
 and the conformal principal components for calculations in the tails. Simulate the Markov chain, and, at each time step and the
 current realization of the state of the chain, simulate the corresponding L\'evy process, and add the simulated increment.
 The most time consuming part (simulation of increments of the L\'evy processes) can be easily parallelized after
 a sample path of the modulating chain is simulated. In the end, the simulated increments are added one by one,
 and we obtain a sample path of the pair (the modulating Markov chain, the Markov-modulated L\'evy process). 
 
\subsection{Monte-Carlo simulations in  the Heston model}\label{ss:MCHeston}
 The most straightforward way is to approximate the volatility process (the square root process) with a Markov chain, and
 apply the scheme outlined in Section \ref{ss:MCRSwLevy}. If this approximation is avoided, the accuracy of the simulations
 improves.

Consider simulations with the time step $\tau$. Apparently, it suffices to simulate the sequence 
$(v_{j\tau}, \De X^{v_{j\tau}}_{\tau})_{j=0,1,\ldots}$, where $v_0>0$ is given, and $\De X^{v_{j\tau}}_{\tau}$ is the increment of the L\'evy process
$X^{v_{j\tau}}$ with the characteristic exponent $\psi(v_{j\tau},\xi)$, where $\psi(v,\xi)=-(vB_0(\tau,\xi)+C_0(\tau,\xi))/\sg_0^{2}$ and
functions $B_0$ and $C_0$ are given by
\eqref{defB0}-\eqref{defD01}.

For the evaluation of the cpdf, pdf and the derivative of the pdf of $X^v$, the domain of the analyticity of $\psi(v,\xi)$, the sinh-acceleration parameters and the mesh size $\ze$ can be chosen the same for any $v\ge 0$.
Since $\Re B_0(\tau,\xi)\to +\infty $ as $\xi\to \infty$ in the conus used in the recommendations for the choice
of the parameters in the sinh-acceleration procedure, the number of terms $N=N(v)$ decreases as $v$ increases.
Hence, the same $N$ can be used for all $v$, hence, the same arrays $\vec{\xi}$, $B_0(\tau,\vec{\xi})$ and $C_0(\tau,\vec{\xi})$
(conformal principal componentss) can be used for all $v\ge 0$. To decrease the number of terms, it is advisable to use different $N$ for $v$ in several selected
intervals, e.g., for $v\in [0,0.01]$, $v\in [0.01, 0.05]$, $v\in [0.05, 0.15]$ and $v\in [0.15, +\infty)$.

For each $v$, we suggest to use different sets of principal components for the evaluation of pdf and cpdf
in a small neighborhood $[-a,a]$ of 0, where, e.g., $a=0.03-0.05$, in the left tail $(-\infty, -a)$ of the distribution, and the right tail 
$(a,+\infty)$. The parameters of the sinh-acceleration are defined by $z_t=0, z_t=-a$ and $z_t=a$, as in Section \ref{EuroHeston}.
Note that, for typical parameters of the Heston model, the arrays of principal components needed to calculate quantiles with the accuracy of the order of E-08
or even E-09 are fairly short:
10-40 in length, and the total number of arrays is $4\times 3\times 3=36$ is moderate as well. Hence, all necessary arrays of the 
conformal principal 
components can be easily calculated and stored at the preliminary step of the simulation procedure. 
Thus, we suggest the following procedure. 
\begin{enumerate}[I.]
\item
(Preliminary step). For the given set of parameters of the Heston model
and the chosen error tolerance for calculation of pdf and cpdf
 \begin{enumerate}[a.]
 \item
Calculate and store arrays of the conformal principal components.
\item
Design a function which, for any $v\ge 0$, approximately calculates $x_-(v)$ and $x_+(v)$ such that
\begin{itemize}
\item
the cpdf $F(v,\cdot)$ is convex on $(-\infty,x_-(v))$ and on $(x_+(v),+\infty)$ or
\item
if, in the numerical procedure for the quantile evaluation, the Newton method is applied to the equation $\ln F(v,x)-\ln A=0$ instead
of $F(v,x)=A$, then $\ln F(v,\cdot)$ must be convex on $(-\infty,x_-(v))$ and on $(x_+(v),+\infty)$.
\end{itemize}
Note that the efficiency of the algorithm decreases only insignificantly if the chosen $x_\pm (v)$ is inside the neighborhood of $\pm\infty$
where $F(v,\cdot)$ (or $\ln F(v,\cdot)$) is convex. Moderately accurate approximations are sufficient.
\end{enumerate}

\item
Simulate a sample path $(v_j)_{j=0,1,\ldots, N}$ of the square root process, with the time increment $\tau$, where $N=T/\tau$ is the number of time steps;
the simulation procedure must produce non-negative numbers.
\item
For each $j=0,1,\ldots, N-1$ do the following (this step can be easily parallelized):
\begin{enumerate}[a.]
\item
take a random sample $a_j$ from the uniform distribution;
\item
using an appropriate set of the conformal principal components (the choice
is determined by the pair $(v_j, a_j)$), find the quantile $Z_j=Z(v_j, a_j)$ of $X^{v_j}_\tau$.
If $a_j\le F(v_j, x_-(v_j))$ or $a_j\ge F(v_j,x_+(v_j))$, use the Newton method with the initial approximation $x_\pm(v_j)$.
Otherwise, apply the bisection method on the interval $[F(v_j, x_-(v_j)), F(v_j, x_+(v_j))]$.
\item
find $Y_j$ from the equation $-Z_j=-Y_j-(\rho/\sg_0)v_j+\mu_0\tau$.
\end{enumerate} 
\item
Calculate the sample path $X_n=X_0+\sum_{j=0}^{n-1}Y_j, n=1,2,\ldots, N,$ of $X=\ln S$.
\item
Repeat Steps II-IV, and use the simulated paths $(v_n,X_n)_{n=0,1,\ldots, N}$ for pricing contingent claims.
\end{enumerate} 
\begin{rem}{\rm If the parallelization at  Step 3b is efficient, the total CPU time for the simulation of one path is, essentially, the sum of the CPU time
needed to simulate a path of the square root process, and of the CPU time needed to calculate the quantile for a given pair $(v_j,a_j)$. Typically, the
latter time is less than 0.1 msec., in  the MATLAB realization.}
\end{rem}

\section{Conclusion}\label{concl}
In the paper, we developed a general methodology for fast and accurate evaluation of integrals of the form
\[
I=\int_{\Im\xi=\om_0}g(\xi)d\xi,
\]
that
appear in many problems in probability, mathematical finance, and other areas of applied mathematics,
and formalized the properties of the integrands that can be calculated using this scheme.
The methodology is applicable if an  integrand $g(\xi)$ admits  analytic continuation to a union of a strip around the line of integration
and a conus that contains the strip, and decays sufficiently fast as $\xi\to \infty$ remaining in the union. The analyticity of the integrand
in the strip and sufficiently fast decay at infinity  allows one to  exploit an important property of the infinite
trapezoid rule, namely, exponential decay of the discretization error as function of $1/\ze$, where $\ze>0$ is the mesh size.
This property is well-known and widely used in the literature. In probability, the characteristic functions of various probability
distributions related to diffusion processes and jump-diffusion processes with exponentially decaying densities of jumps are analytic
in a strip around the real axis. Unfortunately, in many cases of interest such as the CIR model, VG model and KoBoL, the
characteristic function decays slowly as $\xi\to\infty$, and millions of terms in the simplified trapezoid rule may be needed
to satisfy even a moderate error tolerance.

However, if $g(\xi)$ admits analytic continuation to a conus and decays polynomially or exponentially as $\xi\to\infty$ remaining in the conus,
then a change of the variable of the form $\xi=i\om_1+b\sinh (i\om+y)$ in the integral is justified. After the change of variables, the new integrand is analytic
in a strip around the real axis and decays exponentially if the initial integrand decayed polynomially and
as $\exp[-c\exp |y|]$, where $c>0$, if the initial integral decayed exponentially. In the result, in many cases, $N<10$ suffice to satisfy the error tolerance $\eps=10^{-7}$; typically, less than 50 terms suffice,
and in essentially all cases of interest, $N$ of the order of 100-150
suffices to satisfy the error tolerance $10^{-12}$.


We formalized the properties of the characteristic functions of processes and distributions that allow one to apply
the sinh-acceleration, and illustrated the general scheme of the sinh-acceleration with several typical examples: pdf of L\'evy processes;
pricing of European options in L\'evy models, Heston model, CIR model, and a subordinated L\'evy model.
The scheme admits straightforward modification to affine stochastic volatility models and interest rate models (it suffices
to replace in \cite{pitfalls} the fractional-parabolic change of variables $\xi=i\om_1\pm i\sg(1\mp i\eta)^\al$ with the sinh-acceleration,
and take into account that the maximal conus of analyticity is, in the general case, narrower than in the case of the Heston model
and CIR model); jumps can be included as in \cite{pitfalls}. Note that if the fractional-parabolic change of variables is used,
then the rate of decay of the integrand increases but the resulting number of terms remains too large in a number of important cases such as
the evaluation of the probability distribution function at the peak (see \cite{iFT,pitfalls}), and pricing options in the interest models
of the CIR-type.

We also outlined applications of the sinh-acceleration to the calculation of quantiles and Monte-Carlo simulations in L\'evy models,
regime-switching  L\'evy models, and the Heston model. We note that for the evaluation of pdf and cpdf over a long interval $(a,b)$ (even semi-infinite one),
it suffices to evaluate several functions at points of a grid of a small and moderate length,
and use these arrays (we suggest the name {\em the conformal principal components}) for any $x\in (a,b)$. Assuming that the 
conformal principal components are precalculated, the last step requires several microseconds (in MATLAB realization),
hence, quantiles can be calculated in 2-5 dozen of microseconds. It is important that, in the process of simulations,
the truncation of the state space becomes unnecessary, and the truncation errors are avoided.

An additional advantage of the sinh-acceleration as compared to the fractional-parabolic change of variables is that
the width of the initial strip of analyticity is almost irrelevant in the former case as explained in Remark \ref{b-rem} whereas
in the latter case, a narrow strip implies a very large number of terms, and makes it necessary to move the line of integration
to a wider strip \cite{iFT,paraHeston,pitfalls}. However, the angle between the rays that define the conus of analyticity is important.

\sbr
The general scheme of the sinh-acceleration consists of the following steps
\begin{enumerate}[I.]
\item
Find $\gam\le 0 <\gap$ or $\gam<0\le \gap$ such that the integrand $g(\xi)$ is analytic in the cone $\cC_{\gam,\gap}$, and
decays as $\xi\to\infty$ remaining in the cone.
\item
 Set $\om=(\gap+\gam)/2$, $d_0=(\gap-\gam)/2$.
\item
Find a strip $S_{(-\mum,\mup)}$ of analyticity of the integrand around the initial line $\Im\xi=\om_0$ of integration.
\item
Set $a_-= \sin (\min\{\pi/2,-\gam\}), a_+=\sin (\min\{\pi/2,\gap\})$, and
\[
\om_1=\frac{\mup a_-+\mum a_+}{a_++a_-},\
b_0=\frac{\mup-\mum}{a_++a_-}.
\]
\item
Choose $k_b=0.8-0.95, k_d=0.8-0.95$ and set $b=k_bb_0$, $d=k_dd_0$.
\item
Derive an upper bound for the Hardy norm $H$ of $f(y)=g(i\om_1+b\sinh (i\om+y))b\cos(i\om+y)$ as an analytic function
on $S_{(-d,d)}$. Typically, a simple approximation $H=10(|f(id)|+|f(-id)|)$ works well.
\item
Given the error tolerance $\eps$, choose the mesh size as $\ze=2\pi d/\ln(H/\eps)$.
\item
Derive an approximate bound for $g(e^{i\om}\rho)$ and $g(e^{i(\pi-\om)}\rho)$ for $\rho$ in a neighborhood of $+\infty$.
\item
Given the error tolerance $\eps$, use the bound to find $\La_1$ such that
\[
\int_{\La_1}^{+\infty} |g(e^{i\om}\rho)|d\rho+\int_{\La_1}^{+\infty} |g(e^{i(\pi-\om)}\rho)|d\rho<\eps.
\]
\item
Set $\La=\ln(2\La_1/b)$, $N=\mathrm{ceil} (\La/\ze)$.
\item
Apply the simplified trapezoid rule
\begin{equation}\label{simpltrapgen}
\int_{\Im\xi=\om_0}g(\xi)d\xi\approx b\ze\sum_{|j|\le N}g(i\om_1+b\sinh (i\om+j\ze))\cos(i\om+j\ze).
\end{equation}
\item
If $\overline {g(\xi)}=g(-\xi),\ \forall\ \xi$, use the following faster version of \eqref{simpltrapgen}
\begin{equation}\label{simpltrapgen2}
\int_{\Im\xi=\om_0}g(\xi)d\xi\approx 2b\ze\sum_{0\le j\le N}(1-\de_{j0}/2)\Re \left(g(i\om_1+b\sinh (i\om+j\ze))\cos(i\om+j\ze)\right),
\end{equation}
where $\de_{jk}$ is Kronecker's delta.
\end{enumerate}
Note that, in addition to the theoretical bounds for the error of the sinh-acceleration, one can easily check the accuracy
of the result choosing a different pair $(\gam,\gap)$ so that the new $\om$ is different from the old one, and 
a longer and finer grid than recommended, and
recalculate the integral. The probability of a random agreement between
the two results is negligible, hence, the absolute value of the difference is a good proxy for the error.

The sinh-acceleration change of variables can be applied to pricing basket options, European options in quadratic term structure models,
models with Wishart dynamics (in both cases, the integrands decay very slowly, as in the CIR model, hence, accurate calculations
using the popular FFT techniques are essentially impossible as demonstrated in \cite{pitfalls} for affine models of $A_n(n)$ class),
3/2 models and, essentially, any model where the (conditional) characteristic function can be calculated, e.g., Barndorff-Nielsen and Shephard model, and subordinated
models more general than the model considered in the paper.
The methodology can be also applied
to evaluation of special functions \cite{Sinh}, the Wiener-Hopf factors, 
calculation of distributions of the infimum and supremum of L\'evy processes, pricing of path-dependent options,
Monte-Carlo simulations of options with barrier features, pricing in models 
 of Ornstein-Uhlenbeck type, and in many other cases.
The efficiency of the calibration procedure of the Heston model
in \cite{HestonCalibMarcoMe,HestonCalibMarcoMeRisk} can also be improved. To apply the sinh-acceleration to pricing in
regime-switching models,
it suffices to use matrix operations instead of the scalar ones (and, naturally, study the region where the matrix functions
and their reciprocals are analytic; formally, the scheme remains the same). 
Applications to stochastic covariance models are similar to
applications to stochastic volatility models.

\appendix

\section{Numerical examples}\label{numer}
The calculations in the paper
were performed in MATLAB 2017b-academic use, on a MacPro with a
2.8 GHz Intel Core i7 and 16 GB 2133 MHz LPDDR3 RAM.

\subsection{Tables I. Pdf of NTS}\label{tablespdfNTS}
The parameters of the process are  $\mu=0$, $\al=10, \be=0$
for $t=0.004$; $\de=m_2\la^{\nu-2}$, where $m_2=\psi^{\prime\prime}(0)=0.1$ is the second instantaneous moment.
In Table 1, $\nu$ varies, and the pdf is calculated at the peak. In Table 2, $\nu=0.3$ is fixed,
and $x$ varies.

The benchmark prices are obtained using the sinh-acceleration with different $\gap, \gam$, $\ze$ and $N$;
the results differ by less than E-15.
For each $\nu$, $x'=x-\mu t$ and the method
of integration, the mesh size $\ze$ and $\La$ are chosen using the universal prescriptions for the error tolerance
$\eps$.
In some cases, these prescriptions are either inaccurate or lead to the overkill; then
we show the results obtained with $\ze/k_\ze$ and $k_\La\La$ instead of the prescribed $\ze$ and $\La$.
Typically, approximate bounds for the Hardy norm are inaccurate for $\nu<1$ ($\ze$ must be about 30\% smaller)
and lead to an overkill for $\nu>1$ ($\ze$ can be about 5-10\% larger). In some cases, $\La$ can be 5-10\% smaller as well.
The CPU time is in microseconds, the average over 1 mln runs.
\begin{table}

\caption{\small Pdf of $X_t$ at the peak at 0, rounded, and truncation errors of the calculation
using sinh-acceleration and flat inverse Fourier transform. Dependence on the order $\nu$.}
{\tiny
\begin{tabular}{c|ccccccc}
\hline\hline
$\nu$  & 0.1 &	0.3	& 0.5 &	0.9 &	1.1	& 1.5	& 1.9\\
$p_t(0)$ & 1.64335E+11 &	27813.7583 &	1077.36380	& 111.103247 &	64.5381220	& 32.7368302 &	21.6193636
\\\hline
SINH  & & & & & & & \\
\hline
& $\eps=10^{-15}$ & $k_\ze=1$ & $k_\La=1$ &  & & &  \\
$N$ & 30 &	30	& 33 &	32 &	33	& 34 &	35
\\
Error 	& 0	& 0	 &	0 &	0	& 0	& 0 &
0\\
Time & 6.8&	6.7 & 7.0 &	7.0	& 7.0 &	7.1 &	 7.4\\
\hline
& $\eps=10^{-15}$ & $k_\ze=1$ & $k_\La=0.95$ &  & & &  \\
$N$ & 29	& 29 &	31	& 31	& 31 &	32 &	33
\\
Error 	& 0	& 0	 &	0 &	0	& 0	& -1.2E-12 &	-9.9E-14
\\
Time & 6.6& 	6.6 &	6.7 &	6.8 &	6.8 &	6.8 &	6.9
\\
\hline
& $\eps=10^{-7}$ & $k_\ze=1.1$ & $k_\La=1$ &  & & &  \\
$N$ & 19 &	17 &	17 &	17 &	17 &	18 &	18
\\
Error 	& 1.743E+03 &	2.7E-06 &	5.9E-07	& 2.3E-07 &	-9.6E-08 &	5.3E-09 &	-6.6E-09
\\
Time & 5.3 &	4.9 &	5.3 &	5.3&5.3 &	5.2 & 5.5
\\
\hline
& $\eps=10^{-4}$ & $k_\ze=1$ & $k_\La=1$ &  & & &  \\
$N$ & 13 &	10 &	10 &	9 &	9 &	10	& 10
\\
Error 	& -5.1E+06 &	1.48 &	0.020 &	0.0013 &	0.0013 &	-0.00032 &	-8.3549E-05
\\
Re.Err. & -3.1E-05 &	5.3E-05 &	1.9E-05 &	1.2E-05 &	2.0E-05	& -9.8E-06 &	-3.9E-06\\
Time & 4.6 &	4.2 &	4.1 &	4.1 &	4.0 &	4.2 &	4.1
\\\hline
Fract.  & Parabolic  & & & & & & \\
\hline
& $\eps=10^{-15}$ & $k_\ze=1$ & $k_\La=1$ &  & & &  \\
$N$ & 17851 &	10866 &	10244 &	1250 &	729	& 345 &	201
\\
Error 	& -146.4 &	4.3E-09 &	-1.0E-11 &	0 &	0 &	0 &	0
\\
Time & 2161 &	1322 &	1129 &	167 &	112 &	68 &	80
\\
\hline
& $\eps=10^{-7}$ & $k_\ze=1$ & $k_\La=1$ &  & & &  \\
$N$ & 6512 &	3361 &	2921 &	460	& 268 &	130 &	78
\\
Error 	& -113 &	8.0E-08	& 5.7E-10 &	-9.0E-11 &	-2.3E-10 &	-3.5E-10 &	-9.6E-11
\\
Time & 754 &	410 &	350 &	82 &	57 &	55	& 35
\\\hline
& $\eps=10^{-4}$ & $k_\ze=1$ & $k_\La=1$ &  & & &  \\
$N$ & 3334 &	1558 &	1279 &	238	& 138	& 68 &	41
\\
Error 	& -245 &	-0.0003 &	2.9E-05 &	1.3E-07	& 1.4E-06 &	1.7E-07	& -6.0E-07
\\
Time & 390 &	203 &	167 &	92 &	55 &	29 &	19
\\
\hline
Errors & of  flat IFT;& $\ze$  is& fixed & & & &\\
$N=10^5$ &	-1.64E+11 &	-10795 &	-1.5E-04
&	1.1E-07	& 7.7E-08 & 3.1E-08 &	4.1E-09\\
$N=10^6$ &	-1.64E+11 &	-1175 &	2.1E-07
&	1.1E-07	& 7.7E-08 & 3.1E-08 &	4.1E-09\\
$N=10^7$ &	-1.64E+11 &	4.6 &	2.1E-07
&	1.1E-07	& 7.7E-08 & 3.1E-08 &	4.1E-09\\
$N=2\cdot 10^7$ &	-1.64E+11 &	0.28 &	2.1E-07
&	1.1E-07	& 7.7E-08 & 3.1E-08 &	4.1E-09\\\hline
\end{tabular}
}
\begin{flushleft}
{\tiny
$X$: completely symmetric NTS L\'evy process with
$\la=10$, $m_2=\psi^{\prime\prime}(0)=0.1$, $\de=m_2\la^{\nu-2}$, $t=0.004$, $\nu$ varies.
Study of the efficiency of the universal recommendations for the parameter choice
for the sinh-acceleration and fractional-parabolic transformation.
For flat iFT, study of  the dependence of the truncation error on $\nu$ and the number of terms $N$.
Time: CPU time in microseconds, the average over 1 mln runs.
}
\end{flushleft}
 \end{table}

\begin{table}
 \caption{\small Left tail of pdf of $X_t$, rounded, and truncation errors of the calculation
using the sinh-acceleration and flat inverse Fourier transform. Dependence on the distance from the peak.}
{\tiny
\begin{tabular}{ccccccccc}
\hline\hline
$x$ & -0.3 &	-0.25	& -0.2	& -0.15	& -0.1 &	-0.05 &	-0.02 &	-0.01 \\
$p_t(x)$ & 0.0029428&0.0059872	& 0.01277601 &	0.0294055 &	0.0777612 &	0.2894651 &	1.160531 &	2.93835839	\\
SINH  & & & & & & & \\
\hline
& $\eps=10^{-15}$ & $k_\ze=1$ & $k_\La=1$ &  & & &  \\
$N$ & 19 &	20&22 &	24 &	26	& 31 &	37 &	42

\\
Error &	7.0E-17 &	6.9E-17	& -3.0E-16	&-1.0E-16 &	-4.0E-16 &	-1.0E-15&	0	&0
\\
Time & 10.9 &	11.1 &	11.7 &	12.2 &	12.5&	15.6 &	17.6 &	18.1
\\
\hline
& $\eps=10^{-7}$ & $k_\ze=1$ & $k_\La=1$ &  & & &  \\
$N$ & 8	& 8	& 9	& 10 &	11	& 14	& 17 &	19
\\
Error &	1.5E-08 &	-1.2E-08 &	9.1E-09	& 2.1E-08 &	-6.1E-09 &	-4.2E-08 &	-4.6E-08 &	1.1E-07

\\
Time & 7.8 &	7.9 &	8.1 &	8.0 &	8.1 &	8.9 &	9.7 &	10.1
\\
\hline
& $\eps=10^{-4}$ & $k_\ze=1$ & $k_\La=1$ &  & & &  \\
$N$ & 4 &	4 &	5 &	5	& 6 &	8 &	10 &	12
\\
Error &	-1.4E-05 &	2.9E-05 &	-2.6E-05 &	3.1E-05 &	-1.9E-05 &	5.8E-05 &	1.1E-04 &	-1.6E-04
\\
Time & 6.7 &	6.3 &	6.7& 6.7 &	10.2 &	7.1 &	8.2 &	8.9
\\\hline
Fract.  & Parabolic  & & & & & & \\
\hline
& $\eps=10^{-15}$ & $k_\ze=1$ & $k_\La=0.8$ &  & & &  \\
$N$ & 38 &	41 &	44	& 48 &	55	& 70	& 96 &	122
\\
Error 	& 5.0E-17 &	4.0E-17 &	-1.0E-16 &	0 &	1.9E-16	& 1.0E-15 &	0 &	0
\\
Time & 18.0 &	19.3& 22.5 &	23.8 &	26.5 &	33.3 &	44.9 &	56.8
\\\hline
& $\eps=10^{-7}$ & $k_\ze=1$ & $k_\La=0.8$ &  & & &  \\
$N$ & 14 &	15 &	16 &	17 &	20 &	25	& 35 &	44

\\
Error 	& 1.0E-10 &	5.7E-11 &	5.8E-11 &	-1.4E-11 &	8.9E-12 &	-4.6E-11 &	9.0E-09 &	3.2E-07

\\
Time & 9.8 &	10.7 &	11.1 &	11.7 &	12.5&	15.6 &	20.3 &	24.5
\\\hline

& $\eps=10^{-4}$ & $k_\ze=1$ & $k_\La=0.8$ &  & & &  \\
$N$ & 9	& 10 &	10 &	11 &	13	& 16 &	22 &	28
\\
Error 	& 5.5E-06 &	1.9E-06 &	9.0E-07 &	-5.0E-07 &	-9.4E-07 &	3.6E-07 &	-1.3E-07 &	6.5E-07
\\
Time & 6.8 &	7.1 &	6.8 &	7.3&	7.9 &	9.4 &	11.5 &	15.0

\\\hline
& $\eps=10^{-4}$ & $k_\ze=0.95$ & $k_\La=0.8$ &  & & &  \\
$N$ & 9 &	9&10	& 11 &	12 &	15	& 21	& 27
\\
Error 	& 1.7E-05 &	6.5E-06 &	1.2E-06	& 1.6E-06 &	2.4E-06	& -3.5E-06	&-8.5E-07	&-3.5E-07

\\
Time & 6.6 &	6.2&	6.9&	6.9 &	7.3 &	8.5 &11.3 &	13.7
\\\hline

Errors & of  flat IFT;& $\ze$  is& fixed & & & &\\
$N=10^5$ & 0.0057 &	-0.0056 &	0.0056 &	-0.0055 &	0.0054 &	-0.0054 &	-0.48 &	1.53
\\
$N=10^6$ & 0.0018 &	0.00040 &	-0.0018 &	-0.0045 &	-0.0070 &	-0.0088 &	-0.0094 &	-0.0094 \\

$N=10^7$ & -1.3E-06 &	4.6E-06 &	7.1E-06	& 1.2E-06 &	-1.2E-05 &	-2.6E-05 &	-1.4E-05 &	-0.00015 \\


\hline
\end{tabular}
}
\begin{flushleft}
{\tiny
$X$: completely symmetric NTS L\'evy process
of finite variation, with
$\la=10$, $m_2=\psi^{\prime\prime}(0)=0.1$, $\nu=0.3$, $\de=m_2\la^{\nu-2}$, $t=0.004$, $x$ varies.
Study of the efficiency of the universal recommendations for the parameter choice
for the sinh-acceleration and fractional-parabolic transformation.
For flat iFT, study of  the dependence of the truncation error on $\nu$ and the number of terms $N$.
Time: CPU time in microseconds, the average over 1 mln runs.
}

\end{flushleft}
 \end{table}

\subsection{Tables II. The Heston model}\label{tableHeston}
 Table 3: $T=0.004$, comparison of the sinh-acceleration with the fractional-parabolic method, for one strike.

 Table 4:  $T=0.004$,  calculation using the same set of parameters of SINH for all strikes.
 The errors and times for calculation of prices for different numbers of strikes.

 Tables 5-8: the same as Table 4, for $T=0.1, 1.0, 5, 15$.

 Table 9: the comparison of the performance of the sinh-acceleration method with the Lewis-Lipton and Carr-Madan realizations of the flat iFT method.
 In all cases, the standard prescriptions ($\ze=0.125$, $N=4096$) imply negligible  truncation errors, hence,
the errors shown are, essentially, the
discretization errors.
 \begin{table}
\caption{Put in the Heston model. SINH acceleration vs Fractional Parabolic. }
{\tiny
\begin{tabular}{cccccccc}
\hline\hline
$K$ & 85 &	90	& 95 &	100 &	105	& 110 &	115\\
$x'$ & 0.205437159 &	0.1482787452 &	0.0942115239 &	0.0429182295 &	-0.0058719347 &	-0.0523919503 &	-0.0968437129\\
$V_{\mathrm{put}}$ & 8.75606E-07 & 0.0004112657 &	0.046751956	& 1.0603962422	& 5.0125262734	& 9.991210204 &	 14.9908003682\\
\hline\hline
SINH \\\hline
& $\eps=10^{-12}$ & $k_\ze=1.8$ & $k_\La=1.35$ \\
$\ze$ & 0.135219069 &	0.141969971 &	0.149051905 &	0.156224513 &	0.161505024	& 0.154669039	& 0.148348105
\\
$N$ & 58	& 56 &	54 &	52 &	50 &	53 &	55\\
Error &  -2.98E-12 &	3.95E-12 &	0	& -4.00E-14 &	3.91E-14 &	0 &	9.95E-14\\
Time & 48.7 &	48.1 &	47.1 &	45.6 &	45.2 &	46.8 &	48.2\\\hline
& $\eps=10^{-6}$ & $k_\ze=1.8$ & $k_\La=1.35$ \\
$\ze$ & 0.239504852 &	0.251634664 &	0.264309759 &	0.277086239 &	0.286405666	& 0.274054828 &	0.262614271\\
$N$ & 31 &	30	& 29 &	28 &	27 &	28 &	30\\
Error & -1.41E-07 &	-1.35E-07 &	-1.01E-07 &	-6.55E-10 &	5.52E-09 &	-6.43E-08 &	1.45E-07\\
Time & 36.3& 36.0 &	35.6 &	34.9 &	34.8 &	35.6 &	36.0\\\hline
& $\eps=10^{-2}$ & $k_\ze=1.8$ & $k_\La=1.35$ \\
$\ze$ & 0.450079392 &	0.473594213 &	0.497964069 &	0.522276556 &	0.539648211 &	0.515421352 &
	0.492902392\\
$N$ & 14 &	14	& 13 &	13 &	13 &	13 &	14\\
Error & -2.16E-04 &	-6.09E-04 &	-1.63E-04 &	9.62E-06 &	-1.30E-04 &	-8.23E-04 &	-2.49E-03
\\
Time & 28.9 &	30.4 &	27.6 &	26.6 &	27.0 &	27.4 &	27.0
\\\hline\hline
Fract. Para \\\hline
& $\eps=10^{-12}$ & $k_\ze=1$ & $k_\La=1$ \\
$\ze$ & 0.1501040751 &	0.1498588017 &	0.149627309 &	0.1494081636 &	0.158551587 &	0.1583831405 &	0.1582217381\\
$N$ & 290 &	341 &	420	& 563 &	759 &	498	& 393\\
Error  & -1.28E-13 &	2.80E-14 &	-2.80E-14 &	-1.01E-14 &	-7.02E-14 &	-2.90E-13 &	-2.01E-13\\
Time & 115.0 &	109.2 &	120.3 &	148.5 &	198.5 &	135.3 &	115.0\\\hline
& $\eps=10^{-6}$ & $k_\ze=0.85$ & $k_\La=0.85$ & \\
$\ze$ & 0.3006157713 &	0.299780534 &	0.2989939791 &	0.2982509424 &	0.3306004853 &	0.3299784849 &	0.3293834519\\
$N$ & 92 &	109 &	135 &	182 &	237	& 154 &	121\\
Error  & -4.70E-07 &	-1.07E-06 &	3.05E-07 &	-1.91E-07 &	-1.53E-06 &	-1.55E-06 &	-1.5E-06\\
Time & 87.3 &	55.7 &	60.4 &	70.1 &	81.3 &	65.2 &	60.0\\\hline
& $\eps=10^{-2}$ & $k_\ze=0.85$ & $k_\La=0.85$ & \\
$\ze$ & 0.565 &	0.562 &	0.560 &	0.557 &	0.682 &	0.679 &	0.676\\
$N$ & 35 &	42 &	53 &	72	& 86 &	55 &	43\\
Error & 1.2E-04 &	-4.3E-05 &	1.2E-04 &	-0.0057 &	2.0E-04 &	2.8E-05 &	-1.3E-04\\
Time & 60.2 &	48.5 &	51.5 &	57.8 &	67.1 &	55.5 &	43.8\\\hline\hline

\end{tabular}
}
\begin{flushleft}
{\tiny
{\em Put option parameters:} $r=0.02$, $\de=0$, $T=0.004$, $S=100$.
\sbr
{\em Parameters of the Heston model: $v_0=0.18; \rho=-0.58, \sg_0=2.44, \ka=0.30, m=0.18$.}
\sbr
{\em Time: CPU time in microseconds, the average over 1 mln runs.}
\sbr
{\em Given the error tolerance $\eps$, the parameters
of the schemes are chosen for each point,  using the
universal prescriptions with the corrections
factors $k_b=0.8, k_d=0.8$, $k_\ze, k_\La$.}
}
\end{flushleft}

 \end{table}

 \begin{table}
\caption{\small Put in the Heston model, $T=0.004$. Prices and errors of the SINH-acceleration.
}
{\tiny
\begin{tabular}{cccccccc}
\hline\hline
$K$ & 85 &	90	& 95 &	100 &	105	& 110 &	115\\
$x'$ & 0.205437159 &	0.1482787452 &	0.0942115239 &	0.0429182295 &	-0.0058719347 &	-0.0523919503 &	-0.0968437129\\
$V_{\mathrm{put}}$ & 8.75606E-07 & 0.0004112657 &	0.046751956	& 1.0603962422	& 5.0125262734	& 9.991210204 &	 14.9908003682\\
\hline\hline
 $\eps=10^{-12}$ & 4.26E-14 &	5.68E-14 &	-1.42E-14 &	-1.74E-12	& 6.01E-12 &	2.08E-11 &	2.26E-10\\
 $\eps=10^{-6}$ & 8.65E-10	& 7.82E-09 &	1.91E-08 &	-1.68E-07 &	4.23E-07 &	-6.85E-07 &	3.52E-06\\
$\eps=10^{-2}$ & -4.75E-03 &	1.06E-02 &	-2.25E-02 &	4.14E-02 &	-5.33E-02 &	6.09E-02 &	-6.89E-02
\\\hline
\end{tabular}
}
\begin{flushleft}
{\tiny
{\em Put option parameters:} $r=0.02$, $\de=0$, $T=0.004$, $S=100$.
\sbr
{\em Parameters of the Heston model: $v_0=0.18; \rho=-0.58, \sg_0=2.44, \ka=0.30, m=0.18$.}
\sbr

{\em Time: CPU time in microseconds, the average over 1 mln runs.}
\sbr
Given the error tolerance,
the parameters are chosen the same for strikes in the range $[85, 115]$, using the
universal prescriptions with the corrections
factors $k_b=0.8, k_d=0.8$, $k_\ze=1.85, k_\La=1.3$.
\sbr For $\eps=10^{-12}$: $N=89$,	$\ze=0.081939329$, CPU time for 7 and 120 strikes: 69.8 and 586.8 microsec.,
respectively.
\sbr For $\eps=10^{-6}$: $N=57$,	$\ze=0.118334081$, CPU time for 7 and 120 strikes: 45.4 and
 377.6 microsec., respectively.
\sbr For $\eps=10^{-2}$: $N=30$,	$\ze=0.181940202$, CPU time for 7 and 120 strikes: 30.3 and
240.7 microsec., respectively.
}
\end{flushleft}

 \end{table}

 \begin{table}
\caption{\small Put in the Heston model, $T=0.1$. Prices and errors of the SINH-acceleration
with the universal choice of the parameters.
}
{\tiny
\begin{tabular}{cccccccc}
\hline\hline
$K$ & 85 &	90	& 95 &	100 &	105	& 110 &	115\\
$x'$ & 0.2085894213	& 0.1514310075 &	0.0973637862 &	0.0460704918 &	-0.0027196724 &	-0.049239688 &	-0.0936914506\\
$V_{\mathrm{put}}$ & 1.1764633175 &	1.8719759966 &	2.9150895284 &	4.5125209091 &	7.067104472 &	10.7962013124 &	 15.2373482324
\\
\hline\hline
 $\eps=10^{-12}$ & 0 &	4.00E-14	& -3.02E-14	& -7.02E-14	& 0 &	9.95E-14 &	-3.00E-13
\\
 $\eps=10^{-6}$ & -1.44E-10 &	-3.26E-10 &	2.84E-09 &	1.31E-09 &	-5.60E-08 &	2.08E-07 &	-5.02E-07
\\
$\eps=10^{-2}$ & -1.17E-08&	-8.05E-08 &	2.96E-07 &	1.06E-06 &	-1.18E-06 &	-7.91E-06	& -1.70E-05
\\\hline
\end{tabular}
}
\begin{flushleft}
{\tiny
{\em Put option parameters:} $r=0.02$, $\de=0$, $T=0.1$, $S=100$.
\sbr
{\em Parameters of the Heston model: $v_0=0.18; \rho=-0.58, \sg_0=2.44, \ka=0.30, m=0.18$.}
\sbr
{\em Time: CPU time in microseconds, the average over 1 mln runs.}
\sbr
Given the error tolerance,
the parameters are chosen the same for strikes in the range $[85, 115]$, using the
universal prescriptions with the correction
factors $k_b=0.8, k_d=0.8$, $k_\ze=1.85, k_\La=1.3$.
\sbr For $\eps=10^{-12}$: $N=94$,	$\ze=0.080430727
$, CPU time for 7 and 120 strikes: 72.0 and 565.7 microsec.,
respectively.
\sbr For $\eps=10^{-6}$: $N=48$,	$\ze=0.13343117
$, CPU time for 7 and 120 strikes: 48.3 and
 364.3 microsec., respectively.
\sbr For $\eps=10^{-2}$: $N=30$,	$\ze=0.181940202$, CPU time for 7 and 120 strikes: 27.8 and
 237.4 msec., respectively.
}
\end{flushleft}

 \end{table}

\begin{table}
\caption{\small Put in the Heston model, $T=1$. Prices and errors of the SINH-acceleration
with the universal choice of the parameters.
}
{\tiny
\begin{tabular}{cccccccc}
\hline\hline
$K$ & 85 &	90	& 95 &	100 &	105	& 110 &	115\\
$x'$ & 0.2381418803 &	0.1809834665 &	0.1269162452 &	0.0756229508 &	0.0268327867 &	-0.019687229	& -0.0641389916
\\
$V_{\mathrm{put}}$ & 4.7941827931 &	5.6161173264	& 6.646714606 &	8.0122168751 &	9.9462613433 &	12.730505446 &	 16.3323981366
\\
\hline\hline
 $\eps=10^{-12}$ & -1.96E-14 &	1.95E-14 &	9.95E-14 &	0 &	-6.04E-14 &	9.95E-14 &	1.00E-12
\\
 $\eps=10^{-6}$ & 8.24E-10 &	1.41E-09 &	-1.36E-08 &	-5.52E-08 &	-5.51E-08 &	2.74E-07 &	1.06E-06
\\
$\eps=10^{-2}$ & -1.31E-04 &	-1.11E-04 &	8.73E-04 &	-2.28E-03 &	3.50E-03&	-2.14E-03&
	-3.56E-03

\\\hline
\end{tabular}
}
\begin{flushleft}
{\tiny
{\em Put option parameters:} $r=0.02$, $\de=0$, $T=1$, $S=100$.
\sbr
{\em Parameters of the Heston model: $v_0=0.18; \rho=-0.58, \sg_0=2.44, \ka=0.30, m=0.18$.}
\sbr
{\em Time: CPU time in microseconds, the average over 1 mln runs.}
\sbr
Given the error tolerance,
the parameters are chosen the same for strikes in the range $[85, 115]$, using the
universal prescriptions with the correction
factors $k_b=0.8, k_d=0.8$, $k_\ze=1.85, k_\La=1.3$.
\sbr For $\eps=10^{-12}$: $N=85$,	$\ze=0.085671285
$, CPU time for 7 and 120 strikes: 66.0 and 518.9 microsec.,
respectively.
\sbr For $\eps=10^{-6}$: $N=49$,	$\ze=0.130631744$, CPU time for 7 and 120 strikes: 42.6 and
 350.6 microsec., respectively.
\sbr For $\eps=10^{-2}$: $N=26$,	$\ze=0.200931104
$, CPU time for 7 and 120 strikes: 28.8 and
234.9 microsec., respectively.
}
\end{flushleft}

 \end{table}

 \begin{table}
\caption{\small Put in the Heston model, $T=5$. Prices and errors of the SINH-acceleration
with the universal choice of the parameters.
}
{\tiny
\begin{tabular}{cccccccc}
\hline\hline
$K$ & 90 &	100	& 110 &	120 &	130 &	140 &	150
\\
$x'$ & 0.3123277288 &	0.2069672131 &	0.1116570333 &	0.0246456563 &	-0.0553970514 &	-0.1295050235 &	-0.198497895
\\
$V_{\mathrm{put}}$ & 8.9118170191 &	11.3017608315 &	14.4866039624 &	18.9062479333 &	24.8561314222 &	32.0308080039 &	 39.9171298805
\\
\hline\hline
 $\eps=10^{-12}$ & -6.04E-14 &	2.01E-13 &	-3.00E-13 &
 	-7.00E-13 &	-1.50E-12 &	-4.00E-12 &	2.20E-12\\
 $\eps=10^{-6}$ & -2.70E-07 &	1.83E-09 &	4.67E-07 &	8.44E-07 &	2.73E-07 &	-2.09E-06 &	-2.78411E-06
\\
$\eps=10^{-2}$ & 5.00E-03&	-8.27E-04 &	-9.73E-03 & 1.64E-02 &	-2.89E-03&	-1.82E-02&	7.70E-03

\\\hline
\end{tabular}
}
\begin{flushleft}
{\tiny
{\em Put option parameters:} $r=0.02$, $\de=0$, $T=5$, $S=100$.
\sbr
{\em Parameters of the Heston model: $v_0=0.18; \rho=-0.58, \sg_0=2.44, \ka=0.30, m=0.18$.}
\sbr
{\em Time: CPU time in microseconds, the average over 1 mln runs.}
\sbr
Given the error tolerance,
the parameters are chosen the same for strikes in the range $[90, 120]$, using the
universal prescriptions with the correction
factors $k_b=0.8, k_d=0.8$, $k_\ze=1.85, k_\La=1.3$.
\sbr For $\eps=10^{-12}$: $N=75$,	$\ze=0.087187403
$, CPU time for 7 and 120 strikes: 60.2 and 460.0 microsec.,
respectively.
\sbr For $\eps=10^{-6}$: $N=43$,	$\ze=0.13264446
$, CPU time for 7 and 120 strikes: 39.6 and
 338.6 microsec., respectively.
\sbr For $\eps=10^{-2}$: $N=22$,	$\ze=0.203311839
$, CPU time for 7 and 120 strikes: 27.5 and
230.0 microsec., respectively.
}
\end{flushleft}

 \end{table}

 \begin{table}
\caption{\small Put in the Heston model, $T=15$. Prices and errors of the SINH-acceleration
with the universal choice of the parameters.
}
{\tiny
\begin{tabular}{cccccccc}
\hline\hline
$K$ & 90 &	100	& 110 &	120 &	130 &	140 &	150
\\
$x'$ & 0.6406883845 &	0.5353278689 &	0.440017689	& 0.3530063121 &	0.2729636044	& 0.1988556322 &	0.1298627607\\
$V_{\mathrm{put}}$ & 12.4856557684	& 14.8462073848	& 17.4752559196	& 20.4094193312	& 23.6896491628 &	27.3577089222 &	 31.4493345118
\\
\hline\hline
 $\eps=10^{-12}$ & -3.00E-13 &	-1.40E-12 &	-6.50E-12 &	-4.90E-12 &	3.46E-11 &	-2.20E-11	& -7.48E-11
\\
 $\eps=10^{-6}$ & -1.04E-06	&4.53E-06 &	-9.60E-06 &	7.17E-06 &	1.37E-05 &	-3.16E-05 &	-6.82E-06
\\
$\eps=10^{-2}$ & -0.0164 &	-0.053 &	-0.070 &	-0.051 &	0.00314&	0.0753 &	0.139

\\\hline
\end{tabular}
}
\begin{flushleft}
{\tiny
{\em Put option parameters:} $r=0.02$, $\de=0$, $T=15$, $S=100$.
\sbr
{\em Parameters of the Heston model: $v_0=0.18; \rho=-0.58, \sg_0=2.44, \ka=0.30, m=0.18$.}
\sbr
{\em Time: CPU time in microseconds, the average over 1 mln runs.}
\sbr
Given the error tolerance,
the parameters are chosen the same for strikes in the range $[90, 120]$, using the
universal prescriptions with the correction
factors $k_b=0.8, k_d=0.8$, $k_\ze=1.85, k_\La=1.3$.
\sbr For $\eps=10^{-12}$: $N=58$,	$\ze=0.095602143
$, CPU time for 7 and 120 strikes: 56.3 and 451.0 microsec.,
respectively.
\sbr For $\eps=10^{-6}$: $N=32$,	$\ze=0.145718873
$, CPU time for 7 and 120 strikes: 38.2 and
 298.5 microsec., respectively.
\sbr For $\eps=10^{-2}$: $N=15$,	$\ze=0.224004098
$, CPU time for 7 and 120 strikes: 25.0 and
209.3 microsec., respectively.
}
\end{flushleft}

 \end{table}

 \begin{table}
\caption{\small   Put in the Heston model. Panel A: short and moderate maturities; panel B: long maturities.
 Errors (rounded) of the Lewis-Lipton choice of the line of integration $\om=-0.5$
(LLT: simplified trapezoid rule, LLS: Simpson rule) and of Carr-Madan-Schoutens choice $\om=-1.75$
(CMST: simplified trapezoid rule, CMSS: Simpson rule). In all cases, $\ze=0.125$, $N=4096$, hence, the truncation errors
are negligible, and the errors shown are, essentially, the
discretization errors.}
{\tiny
\begin{tabular}{cccccccc}
\hline\hline
$A$ \\\hline
$K$ & 85 &	90	& 95 &	100 &	105	& 110 &	115\\
\hline\hline
$T=0.004$ \\
$x'$ & 0.205437159 &	0.1482787452 &	0.0942115239 &	0.0429182295 &	-0.0058719347 &	-0.0523919503 &	 -0.0968437129\\\hline
LLT & -2.2504E-09	& -2.310E-09 &	-2.372E-09 &	-2.433E-09 &	-2.4947E-09 &	-2.555E-09 &	-2.615E-09\\
LLS & 2.14E-04 &	-1.90E-04 &	-0.0465 &	-1.060 &	-5.012 &	-9.991 &	-14.99\\\hline
CMST & 2.84E-14	& 5.68E-14	& 4.26E-14 &	-6.02E-14 &	-4.00E-14 &	-6.93E-14 &	0\\
CMSS & -2.17E-07 &	-2.17E-07 &	-2.17E-07 &	-2.17E-07 &	-2.17E-07	& -2.17E-07 &	-2.17E-07\\\hline\hline
$T=0.1$ \\
$x'$ & 0.2085894213 &	0.1514310075 &	0.0973637862 &	0.0460704918 &	-0.0027196724 &	-0.049239688 &	-0.0936914506
\\\hline
LLT & -2.248E-09 &	-2.308E-09	& -2.369E-09 &	-2.430E-09 &	-2.450E-09 &	-2.551E-09 &	-2.613E-09\\
LLS & 2.15E-04 &	2.21E-04 &	2.26E-04 &	2.32E-04 &	2.38E-04 &	2.44E-04	& 2.50E-04\\\hline
CMST & 1.00E-13 &	9.99E-15 &	-7.99E-14 &	-3.02E-14 &	0 &	9.95E-14 &	-9.95E-14\\
CMSS & -2.17E-07 &	-2.17E-07 &	-2.17E-07 &	-2.17E-07 &	-2.17E-07 &	-2.17E-07 &	-2.17E-07\\\hline\hline
$T=1$\\
$x'$ & 0.2381418803 &	0.1809834665 &	0.1269162452 &	0.0756229508 &	0.0268327867 &	-0.019687229 &	-0.0641389916
\\\hline
LLT & -2.229E-09	& -2.290E-09 &	-2.348E-09 &	-2.408E-09 &	-2.468E-09 &	-2.527E-09 &	-2.587E-09
\\
LLS & 2.13E-04	& 2.19E-04 &	2.25E-04 &	2.30E-04 &	2.36 &	2.42 &	2.47\\
\hline
CMST & 3.02E-14 &	1.71E-13 &	-1.09E-13 &	1.71E-13 &	0 &	9.95E-14 &	-9.95E-14\\
CMSS & -2.17E-07 &	-2.17E-07 &	-2.17E-07 &	-2.17E-07 &	-2.17E-07 &	-2.17E-07 &	-2.17E-07\\
\hline\hline
$B$ \\\hline
$K$ & 90	& 100 &	110 &	120 &	130 &	140 &	150
\\

\hline\hline
$T=5$ \\
 $x'$ & 0.3123277288 &	0.2069672131 &	0.1116570333 &	0.0246456563 &	-0.0553970514	& -0.1295050235	& -0.198497895
\\\hline
LLT & -2.206E-09 &	-2.316E-09 &	-2.426E-09 &	-2.536E-09 &	-2.647E-09 &	-2.757E-09 &	-2.867E-09\\
LLS & 2.107E-04 &	2.212E-04 &	2.317E-04 &	2.422E-04 &	2.527 &	2.631E-04 &	2.736

\\\hline
CMST & 1.95E-04 &	1.77E-04 &	1.63E-04 &	1.50E-04 &	1.40E-04 &	1.31E-04 &	1.23E-04\\
CMSS & -0.0027 &	-0.0024 &	-0.0022 &	-0.0021 &	-0.0019 &	-0.0018	& -0.0017\\
\hline\hline
$T=15$ \\
 $x'$ & 0.6406883845 &	0.5353278689 &	0.440017689 &	0.3530063121 &	0.2729636044 &	0.1988556322	& 0.1298627607
\\\hline
LLT & -2.021E-09 &	-2.114E-09 &	-2.204E-09	&-2.294E-09 &	-2.384E-09 &	-2.473E-09	& -2.563E-009\\
LLS & 1.929E-04 &	2.014E-04 &	2.099E-04 &	2.185E-04 &	2.270E-04 &	2.355 &	2.440\\\hline
CMST & 0.126 &	0.116 &	0.108 &	0.101 &	0.0947 &	0.0895 &	0.0848\\
CMSS & 0.0471 &	0.0436 &	0.0407 &	0.0382 &	0.0360 &	0.0341 &	0.0324

\\\hline\hline

\end{tabular}
}
 \end{table}

 \subsection{Call option on the bond in CIR model (Table 10) and in CIR-subordinated NTS model (Table 11)}\label{CIRTable}

\begin{table}
\caption{\small Prices of the call option on bond in CIR model, rounded.
Errors, number of terms and CPU times (in msc) of different realizations of iFT.
}
{\tiny
\begin{tabular}{c|cccccccc}
\hline
$K$ & 97.50512024 &	97.6461914 &	97.78746667 &	97.92894634 &	98.0706307 &	98.21252005 &	98.35461469 &	 98.49691491
\\
$z_{TK}$ & -0.02 &	-0.0175	& -0.015 &	-0.0125 &	-0.01 &	-0.0075 &	-0.005 &	-0.0025
\\\hline
Call 1 & \\
Price & 0.876713465 &	0.756024612	& 0.636971345 &	0.519888515 &	0.40523729 &	0.293696753	&
0.186378527	& 0.08550053\\
$\ze$ & 0.110853 &	0.110844 &	0.110835 &	0.110826 &	0.110817& 0.110808 &	0.110799 &	0.110791\\
$N$ & 43 &	44 &	45 &	46 &	48 &	50	& 53 &	58\\
Time & 20.2 &	20.5 &	20.6 &	20.9 &	21.4 &	21.9 &	22.7 & 24.1
\\\hline
Call 2 &\\
Err & 1.54E-14 &	-5.11E-13 &	-5.46E-13 &	-5.93E-13 &	-6.53E-13	& -6.83E-13	& -7.36E-13	&-7.61E-13\\
$N$ & 58 &	47 &	47 &	48 &	50 &	51	& 53	& 57\\
Time & 24.9 &	25.7 &	25.8 &	26.3 &	26.7 &	26.9 &	27.7 &	28.9\\\hline
FrPara & $\al=2.8$ \\
Err & 1.54E-14 &	-5.11E-13 &	-5.46E-13 &	-5.93E-13 &	-6.64E-13 &	-6.83E-13	& -7.36E-13	& -7.72E-13\\
$N$ & 759 &	795 &	838	& 891 &	961 &	1060 &	1217	&1540\\
Time & 215.4 &	216.0 &	215.0 &	216.2 &	216.9 &	215.4 &	217.5 &	217.8\\\hline
Flat iFT& $N=10^5$ &\\
Err & -5.51E-06 &	-7.33E-06 &	-8.99E-06 &	-1.03E-05 &	-1.11E-05	& -1.10E-05 &	-9.08E-06 &	5.29E-07\\
Time & 4580 &	4536 &	4512 &	4597 &	4764 &	4688 &	4937 &	4581\\\hline
\end{tabular}
}
\begin{flushleft}
{\tiny
{\em Parameters of CIR model:} $\ka=1.6; \theta=0.01, \sg=0.5$.
\sbr
{\em Bond} matures at $T+\tau=3$, spot price 99.384925, implied $r_0=0.01$.
\sbr
{\em Call option} matures at $\tau=1$; strikes $K$ and $z=\ln(C(2,0)-\ln K)/B(2,0)$ are shown
in the table.
\sbr
{\em Call 1:} call option prices calculated using the sinh-acceleration with the parameters chosen
for a curve in the lower half-plane and put-call parity; errors
of these prices is less than $2*E-14$. The mesh and truncation parameters used are
0.9 times larger and 0.95 smaller than the general prescription for the error tolerance $E-13$ recommends.
\sbr
{\em Call 2:} errors of option prices calculated using the sinh-acceleration with the parameters
chosen for the curve that is above $-iB(T-\tau,0)$ w.r.t. to prices Call 1; no put-call parity is needed.
The mesh and truncation parameter are
0.9 times larger and 0.95 smaller than the general prescription for the error tolerance recommends.
\sbr
{\em FrPara:} errors of option prices calculated using the fractional-parabolic method with the parameters
chosen so that the errors are of the same order of magnitude as the ones for Call 1.
\sbr {\em Flat iFT:} errors of option prices calculated using the flat iFT with $N=100,000$ terms.
\sbr
{\em Time:} CPU time in microsec., the average over 1 mln runs.
}
\end{flushleft}
\end{table}

\begin{table}
\caption{\small Prices of the call option in the CIR-subordinated NTS model, rounded.
Errors, meshes,  and CPU times of different realizations of iFT.
}
{\tiny
\begin{tabular}{c|cccccccc}
\hline
$K$ & 97.50512024 &	115.0273799 &	112.7496852	& 110.5170918 &	108.3287068 &	106.1836547 &	104.0810774	& 102.020134
\\
$\ln(S_0/K)$ & -0.14 &	-0.12 &	-0.1 &	-0.08 &	-0.06 &	-0.04 &	-0.02 &	0
\\\hline
$V_{call}$ & 0.000300147 &	0.000359207 &	0.000439001 &	0.000552719 &	0.000728595 &	0.001042077	& 0.001808313	& 0.047791256
\\\hline
SINH1 & $\gam=-\pi/3$ & $\gap=0$ & $\eps=10^{-15}$& \\
$\ze$ & 0.051563 &	0.051557 &	0.051551 &	0.051544 &	0.051538 &	0.051531 &	0.051525 &	0.051518\\
$N$ & 139 &	142 &	145 &	148	& 152 &	157 &	163 &	168 \\
$Err$ & 3.0E-15 &	1.5E-16	& 6.0E-16	& 1.1E-15	& -3.3E-16 &	4.7E-16 &	2.6E-15 &	1.5E-07\\
Time & 118.8 &	120.1 &	121.7 &	123.3 &	125.8 &	129.0 &	132.2 &	147.2
\\\hline
SINH2 & $\gam=-\pi/2$ & $\gap=0$ & $\eps=10^{-15}$&\\
$\ze$ & 0.07762 &	0.07760 &	0.07759 &	0.07757 &	0.07756 &	0.07754 &	0.07752 &	0.07751\\
$N$ & 90 &	93 &	95 &	98 &	101 &	105	& 109 &	115\\
$Err$ & -1.8E-15 &	-1.3E-15 &	3.0E-17 &	-2.1E-15 &	-4.0E-15	&-1.1E-15	 & 3.2E-15 &	6.9E-06\\
Time & 84.4 & 	84.8 &	86.6 & 	88.5 &	90.7 &	92.6 &	95.5 &	158.7
\\\hline
SINH3 & $\gam=-\pi/2$ & $\gap=0$ & $\eps=10^{-7}$& \\
$\ze$ &0.15009 & 	0.15003 &	0.14998 &	0.14992 &	0.14986 &	0.14981 &	0.14975 &	0.14969\\
$N$ & 41 &	42 &	43 &	&	46 &	49 &	52 &	56\\
$Err$ & 5.2E-13 &	5.2E-13 &	5.1E-13 &	5.1E-13 &	5.0E-13	& 4.9E-13	 & 4.9E-13 &	1.8E-04\\
Time & 48.7 &	49.1 &	49.6 &	51.0 &	51.6 &	53.9 &	55.4 &	128.4
\\\hline
Flat iFT & $\om_0=-1.75$ & $ \ze=0.25$ & $N=16384$ & \\
$Err$ & 5.02E-07 &	4.82E-07 &	4.63E-07 &	4.47E-07 &	4.36E-07 &	4.39E-07 &	5.11E-07 &	1.07E-05\\
Time & 2920 &	2893 &	2877 &	2877&	2905&	2880 &	28753 &	2773
\\\hline
\end{tabular}
}
\begin{flushleft}
{\tiny
{\em Parameters of CIR subordinator:} $\ka=1.6; \theta=0.01, \la=0.25$, $y_0=0.02$.
\sbr
{\em Parameters of NTS model:} $m_2=0.1$,	$\nu=1.6$, $\de=0.097$,	$\al=3$, $\be=0$, $\mu=0$.
\sbr
{\em Call option:} maturity  $\tau=0.004$, $r=0.02$; spot $S_0=100$, strikes $K$ and $\ln S_0/K$ are shown
in the table.
\sbr {\em Benchmark prices} calculated using several sets of the parameters of the sinh-acceleration;
errors less than $5E-15$ in the absolute value.
\sbr
{\em SINH1:} $\ze$ (rounded), $N$, errors and CPU time when using the general prescription
with $\gam=-\pi/3, \gap=0$, for $\eps=10^{-15}$.
\sbr
{\em SINH2:} $\ze$ (rounded), $N$, errors and CPU time when using the general prescription
with $\gam=-\pi/2, \gap=0$, for $\eps=10^{-15}$.
\sbr
{\em SINH3:} $\ze$ (rounded), $N$, errors and CPU time when using the general prescription
with $\gam=-\pi/2, \gap=0$, for $\eps=10^{-7}$.
\sbr
In all cases SINH$j$, $j=1,2,3$, $\ze$ is 2/3 of recommended, and $\La=N\ze$ is 1.2 times larger.
\sbr
{\em Flat iFT:} errors and CPU time of the calculation using the flat iFT with the standard prescription
$\om_0=-1.75, \ze=0.25, N=16384$.
\sbr
{\em Time:} CPU time in msc, the average over 100,000 runs.
}
\end{flushleft}
\end{table}
\subsection{Examples of calculation of quantiles (Table 12)}\label{examplesfractiles}
\begin{table}[]
 \caption{\small Errors of approximations L,N,LL, LN and QT,
 for different fractiles $F^{-1}(A)$ and intervals of different length, containing
 $F^{-1}(A)$.}
 {\small
\begin{tabular}{c|cc|cc|cc|cc|cc}
\hline
$A$ & $2*10^{-9}$ & & $e^{-16}$ & & $10^{-5}$ & & 0.001 & & 0.3 & \\\hline
$h$ & 0.01 & 0.001 & 0.01 & 0.001 & 0.01 & 0.001 & $10^{-3}$ & $10^{-4}$ & $10^{-4}$ & $10^{-5}$\\\hline
$L$ & -7.3E-05 &	-6.8E-07 &	-7.7E-05 &	-7.9E-07 &	-9.3E-05 &	-9.4E-07 &	-2.4E-06	& -2.5E-08	& -6.4E-06	& -5.0E-08\\
$N$ & 7.7E-05 &	1.2E-06 &	0.0001 &	8.8E-07	& 0.00015 &	1.4E-06 &	4.3E-06 &	2.0E-08	&
1.7E-06 &	9.5E-09\\
$LL$ & -8.4E-07 &	-7.8E-09 &	-1.7E-06 &	-1.7E-08 &	-5.7E-06 &	-5.8E-08 &	-6.1E-07 &	-6.5E-09 &	-3.6E-06	& -2.9E-08\\
$LN$ & 8.9E-07	& 1.4E-08 &	2.4E-06	& 2.0E-08 &	9.0E-06	& 8.6E-08 &	1.1E-06	& 5.0E-09 &	1.0E-06 &	5.4E-09\\
$QT$ & -1.5E-09 &	-3.4E-12 &	-6.9E-09 &	-7.1E-12 &	-5.7E-08 &	-5.3E-11 &	-4.1E-09 &	-1.3E-12 &	-5.5E-08 &	 -2.1E-11\\\hline
\end{tabular}}
\end{table}


\begin{thebibliography}{}
\bibliographystyle{spphys}


\bibitem{CGMYsim}
L.~Ballotta and I.~Kyriakou.
\newblock Monte carlo simulation of the {C}{G}{M}{Y} process and option
  pricing.
\newblock {\em Journal of Futures Markets}, 34(12):1095--1121, December 2014.

\bibitem{B-N}
O.E. Barndorff-Nielsen.
\newblock Processes of {N}ormal {I}nverse {G}aussian {T}ype.
\newblock {\em Finance and Stochastics}, 2:41--68, 1998.

\bibitem{B-N-L}
O.E. Barndorff-Nielsen and S.Z. Levendorski\v{i}.
\newblock Feller {P}rocesses of {N}ormal {I}nverse {G}aussian type.
\newblock {\em Quantitative Finance}, 1:318--331, 2001.

\bibitem{MCMityaLevy}
M.~Boyarchenko.
\newblock Fast simulation of {L}\'evy processes.
\newblock Working paper, August 2012.
\newblock Available at SSRN: http://ssrn.com/abstract=2138661 or
  http://dx.doi.org/10.2139/ssrn.2138661.

\bibitem{one-sidedCDS}
M.~Boyarchenko and S.~Levendorski\u{i}.
\newblock Ghost {C}alibration and {P}ricing {B}arrier {O}ptions and {C}redit
  {D}efault {S}waps in spectrally one-sided {L}\'evy models: The {P}arabolic
  {L}aplace {I}nversion {M}ethod.
\newblock {\em Quantitative Finance}, 15(3):421--441, 2015.
\newblock Available at SSRN: http://ssrn.com/abstract=2445318.

\bibitem{paraLaplace}
S.~Boyarchenko and S.~Levendorski\u{i}.
\newblock Efficient {L}aplace inversion, {W}iener-{H}opf factorization and
  pricing lookbacks.
\newblock {\em International Journal of Theoretical and Applied Finance},
  16(3):1350011 (40 pages), 2013.
\newblock Available at SSRN: http://ssrn.com/abstract=1979227.

\bibitem{iFT}
S.~Boyarchenko and S.~Levendorski\u{i}.
\newblock Efficient variations of {F}ourier transform in applications to option
  pricing.
\newblock {\em Journal of Computational Finance}, 18(2):57--90, 2014.
\newblock Available at SSRN: http://ssrn.com/abstract=1673034.

\bibitem{BarrStIR}
S.~Boyarchenko and S.~Levendorski\u{i}.
\newblock Efficient pricing barrier options and {C}{D}{S} in {L}\'evy models
  with stochastic interest rate.
\newblock {\em Mathematical Finance}, 2016.
\newblock DOI: 10.1111/mafi.12121. Available at SSRN:
  http://ssrn.com/abstract=2544271.

\bibitem{KoBoL}
S.I. Boyarchenko and S.Z. Levendorski\u{i}.
\newblock Option pricing for truncated {L}\'evy processes.
\newblock {\em International Journal of Theoretical and Applied Finance},
  3(3):549--552, July 2000.

\bibitem{NG-MBS}
S.I. Boyarchenko and S.Z. Levendorski\u{i}.
\newblock {\em Non-{G}aussian {M}erton-{B}lack-{S}choles {T}heory}, volume~9 of
  {\em Adv. Ser. Stat. Sci. Appl. Probab.}
\newblock World Scientific Publishing Co., River Edge, NJ, 2002.

\bibitem{CGMY}
P.~Carr, H.~Geman, D.B. Madan, and M.~Yor.
\newblock The fine structure of asset returns: an empirical investigation.
\newblock {\em Journal of Business}, 75:305--332, 2002.

\bibitem{ChenFengLin}
Z.~Chen, L.~Feng, and X.~Lin.
\newblock Simulation of {L}\'evy processes from their characteristic functions
  and financial applications.
\newblock {\em ACM Transactions on Modeling and Computer Simulation}, 22(3),
  2011.
\newblock Available at: http://ssrn.com/abstract=1983134.

\bibitem{MarcoDiscBarr}
M.~de~Innocentis and S.~Levendorski\u{i}.
\newblock Pricing discrete barrier options and credit default swaps under
  {L}\'evy processes.
\newblock {\em Quantitative Finance}, 14(8):1337--1365, 2014.
\newblock Available at: http://dx.doi.org/10.1080/14697688.2013.826814.

\bibitem{EK}
E.~Eberlein and U.~Keller.
\newblock Hyperbolic distributions in finance.
\newblock {\em Bernoulli}, 1:281--299, 1995.

\bibitem{feng-lin11}
L.~Feng and X.~Lin.
\newblock Inverting analytic characteristic functions with financial
  applications.
\newblock {\em SIAM Journal on Financial Mathematics}, 4(1):372--398, 2011.

\bibitem{glass-liu-07}
P.~Glasserman and Z.~Liu.
\newblock Sensitivity estimates from characteristic functions.
\newblock In {\em Proceedings of the 2007 {W}inter {S}imulation {C}onference'},
  pages 932--940. Institute of Electrical and Electronics Engineers, Inc.,
  Piscataway, New Jersey, 2007.

\bibitem{GlassermanLiu07}
P.~Glasserman and Z.~Liu.
\newblock Estimating {G}reeks in {S}imulating {L}\'evy-{D}riven {M}odels.
\newblock {\em Journal of Computational Finance}, 14:3--56, 2010.

\bibitem{glass-liu-10}
P.~Glasserman and Z.~Liu.
\newblock Sensitivity estimates from characteristic functions.
\newblock {\em Operations Research}, 58(6):1611--1623, 2010.

\bibitem{heston-model}
S.~L. Heston.
\newblock A closed-form solution for options with stochastic volatility with
  applications to bond and currency options.
\newblock {\em The Review of Financial Studies}, 6(2):327–--343, 1993.

\bibitem{HestonCalibMarcoMe}
M.de Innocentis and S.~Levendorski\u{i}.
\newblock Calibration and {B}acktesting of the {H}eston {M}odel for
  {C}ounterparty {C}redit {R}isk.
\newblock Working paper, April 2016.
\newblock Available at SSRN: http://ssrn.com/abstract=2757008.

\bibitem{HestonCalibMarcoMeRisk}
M.de Innocentis and S.~Levendorski\u{i}.
\newblock Calibration {H}eston {M}odel for {C}redit {R}isk.
\newblock {\em Risk}, pages 90--95, September 2017.

\bibitem{kou}
S.G. Kou.
\newblock A jump-diffusion model for option pricing.
\newblock {\em Management Science}, 48(8):1086--1101, August 2002.

\bibitem{KW1}
S.G. Kou and H.~Wang.
\newblock First passage times of a jump diffusion process.
\newblock {\em Adv. Appl. Prob.}, 35(2):504--531, 2003.

\bibitem{KW2}
S.G. Kou and H.~Wang.
\newblock Option pricing under a double exponential jump diffusion model.
\newblock {\em Management Science}, 50(9):1178--1192, September 2004.

\bibitem{beta}
A.~Kuznetsov.
\newblock Wiener-{H}opf factorization and distribution of extrema for a family
  of {L}\'evy processes.
\newblock {\em Ann.Appl.Prob.}, 20(5):1801--1830, 2010.

\bibitem{amer-put-levy-maphysto}
S.~Levendorski\u{i}.
\newblock Pricing of the {A}merican put under {L}\'evy processes.
\newblock Research Report MaPhySto, Aarhus, 2002.
\newblock Available at http://www.maphysto.dk/publications/MPS-RR/2002/44.pdf,
  http://www.maphysto.dk/cgi-bin/gp.cgi?publ=441.

\bibitem{amer-put-levy}
S.~Levendorski\u{i}.
\newblock Pricing of the {A}merican put under {L}\'evy processes.
\newblock {\em International Journal of Theoretical and Applied Finance},
  7(3):303--335, May 2004.

\bibitem{paraHeston}
S.~Levendorski\u{i}.
\newblock Efficient pricing and reliable calibration in the {H}eston model.
\newblock {\em International Journal of Theoretical and Applied Finance},
  15(7), 2012.
\newblock 125050 (44 pages).

\bibitem{Sinh}
S.~Levendorski\u{i}.
\newblock {F}ractional-{P}arabolic {D}eformations with {S}inh-{A}cceleration.
\newblock Working paper, April 2016.
\newblock Available at SSRN: http://ssrn.com/abstract=2758811.

\bibitem{pitfalls}
S.~Levendorski\u{i}.
\newblock Pitfalls of the {F}ourier {T}ransform method in {A}ffine {M}odels,
  and remedies.
\newblock {\em Applied Mathematical Finance}, 23, 2016.
\newblock Avaialble at http://dx.doi.org/10.1080/1350486X.2016.1159918,
  http://ssrn.com/abstract=2367547.

\bibitem{UltraFast}
S.~Levendorski\u{i}.
\newblock Ultra-{F}ast {P}ricing {B}arrier {O}ptions and {C}{D}{S}s.
\newblock {\em International Journal of Theoretical and Applied Finance},
  20(-), 2017.
\newblock Available at SSRN: http://ssrn.com/abstract=2713497 or
  http://dx.doi.org/10.2139/ssrn.2713497.

\bibitem{paired}
S.Z. Levendorski\u{i}.
\newblock Method of paired contours and pricing barrier options and {C}{D}{S}
  of long maturities.
\newblock {\em International Journal of Theoretical and Applied Finance},
  17(5):1--58, 2014.
\newblock 1450033 (58 pages).

\bibitem{lipton-risk}
A.~Lipton.
\newblock Assets with jumps.
\newblock {\em Risk}, pages 149--153, September 2002.

\bibitem{lipton-columbia}
A.~Lipton.
\newblock Path-dependent options on assets with jumps.
\newblock 5$^{\textrm{th}}$ Columbia-Jaffe Conference, April 2002.
\newblock Available at http://www.math.columbia.edu/~lrb/columbia2002.pdf.

\bibitem{Lucic}
V.~Lucic.
\newblock On singularities in the {H}eston model.
\newblock Working paper, 2007.
\newblock Available at http://ssrn.com/abstract=1031222.

\bibitem{MCC98}
D.B. Madan, P.~Carr, and E.C. Chang.
\newblock The {V}ariance {G}amma process and option pricing.
\newblock {\em European Finance Review}, 2:79--105, 1998.

\bibitem{MM91}
D.B. Madan and F.~Milne.
\newblock Option pricing with {V}.{G}. martingale components.
\newblock {\em Mathematical Finance}, 1(4):39--55, 1991.

\bibitem{MS90}
D.B. Madan and E.~Seneta.
\newblock The {V}ariance {G}amma ({V}.{G}.) model for share market returns.
\newblock {\em Journal of Business}, 63:511--524, 1990.

\bibitem{madan-yor}
D.B. Madan and M.~Yor.
\newblock Representing the {CGMY} and {M}eixner {L}\'evy processes as time
  changed {B}rownian motions.
\newblock {\em Journal of Computational Finance}, 12(1):27--47, 2009.

\bibitem{merton-model}
R.C. Merton.
\newblock Option pricing when underlying stock returns are discontinuous.
\newblock {\em Journal of Financial Economics}, 3:125--144, 1976.

\bibitem{Rosinski07}
I.~Rosinski.
\newblock Tempering stable processes.
\newblock {\em Stoch. Proc. and Appl.}, 117:677–707, 2007.

\bibitem{stenger-book}
F.~Stenger.
\newblock {\em Numerical {M}ethods based on {S}inc and {A}nalytic functions}.
\newblock Springer-Verlag, New York, 1993.

\end{thebibliography}
\end{document}